\documentclass[letterpaper]{article} % DO NOT CHANGE THIS
\usepackage{aaai24}  % DO NOT CHANGE THIS
\usepackage{times}  % DO NOT CHANGE THIS
\usepackage{helvet}  % DO NOT CHANGE THIS
\usepackage{courier}  % DO NOT CHANGE THIS
\usepackage[hyphens]{url}  % DO NOT CHANGE THIS
\usepackage{graphicx} % DO NOT CHANGE THIS
\usepackage{natbib}  % DO NOT CHANGE THIS AND DO NOT ADD ANY OPTIONS TO IT
\usepackage{caption} % DO NOT CHANGE THIS AND DO NOT ADD ANY OPTIONS TO IT
\usepackage{url}
\usepackage{amsmath}
\usepackage{amssymb}
\usepackage{mathtools}
\usepackage{amsthm}
\usepackage{amsfonts} 
\usepackage{bm}
\usepackage{multirow}

\def\PTD{\textsc{PTD}\xspace}
\def\OAK{\textsc{OAK}\xspace}
\def\POAK{\textsc{POAK}\xspace}
\def\POAKi{\textsc{POAKi}\xspace}
\def\SAD{\textsc{SAD}\xspace}
\def\BAU{\textsc{BAU}\xspace}
\def\Uniform{\textsc{Uniform}\xspace}
\usepackage[ruled,algo2e]{algorithm2e} % For algorithms

\newcommand{\full}[1]{}
\DeclareMathOperator{\argmax}{argmax}

\def\calX{{\mathcal{X}}}
\def\mat{{\mathcal{M}}}

\theoremstyle{plain}
\newtheorem{theorem}{Theorem}[section]
\newtheorem{proposition}[theorem]{Proposition}
\newtheorem{lemma}[theorem]{Lemma}
\newtheorem{corollary}[theorem]{Corollary}
\theoremstyle{definition}

\theoremstyle{remark}
\newtheorem{remark}[theorem]{Remark}

\usepackage[dvipsnames]{xcolor}
\newcount\Comments  % 0 suppresses notes to selves in text
\definecolor{darkgreen}{rgb}{0,0.6,0}
\definecolor{purple}{rgb}{1,0,1}
\newcommand{\kibitz}[2]{\ifnum\Comments=1{\textcolor{#1}{#2}}\fi}
\newcommand{\rmr}[1]{\kibitz{blue}{[RM:#1]}}
\newcommand{\xc}[1]{\kibitz{red}{[XC: #1]}}

\newcommand{\an}[1]{\kibitz{purple}{[An:#1]}}

\usepackage{algorithm}
\usepackage{newfloat}
\usepackage{listings}
\DeclareCaptionStyle{ruled}{labelfont=normalfont,labelsep=colon,strut=off} % DO NOT CHANGE THIS
\floatstyle{ruled}
\newfloat{listing}{tb}{lst}{}
\floatname{listing}{Listing}
\title{Efficient Online Crowdsourcing with Complex Annotations}
\author{
    Reshef Meir\textsuperscript{\rm 1,}\thanks{Work done while visiting Meta's Central Applied Science},
    Viet-An Nguyen\textsuperscript{\rm 2},
    Xu Chen\textsuperscript{\rm 2},
    Jagdish Ramakrishnan\textsuperscript{\rm 2},
    Udi Weinsberg\textsuperscript{\rm 2}
}
\affiliations{
    \textsuperscript{\rm 1}Technion---Israel Institute of Technology\\
    \texttt{reshefm@dds.technion.ac.il}\\
    \textsuperscript{\rm 2}Central Applied Science, Meta\\
    \texttt{\{vietan,xuchen2,jagram,udi\}@meta.com}
}

\newcommand{\newsec}[1]{\vspace{-0mm}\section{#1}}
\newcommand{\newsubsec}[1]{\vspace{-0mm}\subsection{#1}}
\newcommand{\newpar}[1]{\vspace{-0mm}\paragraph{#1}}

\begin{document}

\maketitle
\begin{abstract}
Crowdsourcing platforms use various truth discovery algorithms to aggregate annotations from multiple labelers. In an online setting, however, the main challenge is to decide whether to ask for more annotations for each item to efficiently trade off cost (i.e., the number of annotations) for quality of the aggregated annotations. In this paper, we propose a novel approach for general complex annotation (such as bounding boxes and taxonomy paths), that works in an online crowdsourcing setting. We prove that the expected average similarity of a labeler is linear in their accuracy \emph{conditional on the reported label}. This enables us to infer reported label accuracy in a broad range of scenarios. We conduct extensive evaluations on real-world crowdsourcing data from Meta and show the effectiveness of our proposed online algorithms in improving the cost-quality trade-off.
\end{abstract}

\newsec{Introduction}
\label{sec:intro}

% introduce crowdsourcing
\textit{Crowdsourcing} refers to a broad collection of cost-efficient methods to acquire information from a large population of non-experts~\cite{DoanCACM2011crowdsourcing}. 
%,ZhengVLDB17,Paun2022statistical}. 
Within crowdsourcing, a common task is to ask \textit{workers} (aka \textit{reviewers} or \textit{labelers}) to provide a correct \textit{annotation} (aka \textit{label}) to a piece of information.
 
% introduce complex annotations
The annotations can be simple such as a yes/no answer to whether there is a car in a given photo, or a real number representing the future price of a commodity. However, in many crowdsourcing tasks, the responses from labelers may comprise of more \textit{complex annotations} such as textual spans, bounding boxes, taxonomy paths, or translations.

% introduce modeling approaches for complex annotations: (1) customize models and (2) generic distance-based models
These annotations are then aggregated to obtain the best possible estimation of some underlying correct answer. One typical approach to aggregate multiple annotations is to identify good labelers based on custom, domain-specific statistical models (see Related Work). % develop custom modeling approaches for each specific task of interest including sequence annotation~\cite{NguyenACL2017}, bounding box and key point~\cite{BransonCVPR2017lean}, and social pairwise preference~\cite{LiWWW2022context}. 
While these specialized models were shown to be effective for their respective tasks, one key limitation for such approach is that a custom algorithm is needed for each type of complex annotation. 
Recent work has explored pairwise similarities between annotations as a general approach for identifying good labelers across different types of complex annotations~\cite{braylan2023general,meir2023easy}. The assumption that good workers are similar to one another in terms of their reported annotations (whereas poor workers do not) is referred to as the \textit{Anna Karenina principle} in \citet{meir2023easy}.

\newsubsec{Challenges and Contributions}
\label{sec:challenges_contributions}
\begin{figure}[t]
    \centering\vspace{-0mm}
    \includegraphics[width=\linewidth]{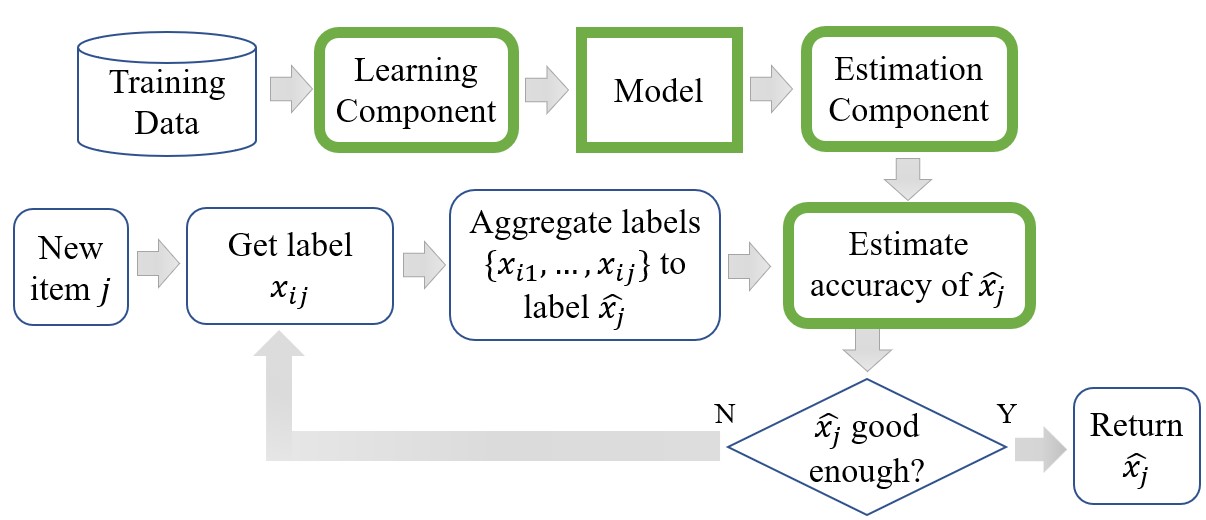}
    \caption{General online crowdsourcing process, in which the components with green solid frame are this work's focus.}
    \label{fig:online_general}\vspace{-0mm}
\end{figure}

% introduce online crowdsourcing and challenges
The focus of this paper is on a prevalent and practical crowdsourcing scenario where on the one hand, annotations can be \textit{complex}, and on the other hand we must decide \textit{on-the-fly} whether we should get an additional annotation for each item. %Figure~\ref{fig:online_general} illustrates this online crowdsourcing setup. 
In this setting, unfortunately, most of the truth discovery algorithms (mentioned above and surveyed in later sections) are inapplicable. We next briefly describe the online crowdsourcing process of interest to articulate the challenge.

\newpar{From offline to online crowdsourcing} The stylized truth discovery model often used in crowdsourcing literature assumes the existence of an \textit{a priori} given set of questions/items to annotate and a fixed group of workers, each assigned to annotate all or some of the items. The designed algorithm takes these observed annotations as offline input and outputs the best aggregated annotation for each item.

However, practical real-time crowdsourcing processes exhibit a distinct setup. As illustrated in Fig.~\ref{fig:online_general}, items arrive sequentially and the system determines when to stop collecting labels for each item. This decision can be informed by the annotations collected so far for this item, and by the knowledge learned from previous annotations for other items used as ``training data''. The CLARA system developed at Meta tackles the same online crowdsourcing problem but only supports categorical labels~\cite{NguyenKDD2020}. \emph{Task assignment}, which selects labelers to annotate items, is an active area of research~\cite{LiTKDE2016crowdsourced,hettiachchicsur2022survey}, but is out of the scope of this work, as it is addressed by a different system. 

% contributions
\newpar{Main contributions} Based on the Anna Karenina (AK) principle as explained above, we suggest an \textit{Online AK algorithm} (\OAK) that estimates the accuracy of each labeler by measuring the average similarity to all others on training data.
OAK is an adaptation of the Proximity-based Truth Discovery (\PTD) algorithm proposed by \citet{meir2023easy} for the online crowdsourcing setting with general complex annotations. The main difference is the way average similarity is \emph{used}, which is to decide on when to stop acquiring labels, rather than how to weigh collected labels. 

Our main contribution is a \textit{Partition-based extension} of \OAK (\POAK).  %By  first partitioning the reported annotations into types; and  then applying the \OAK algorithm to estimate workers' competence per annotation type. 
Although \POAK can be viewed as employing multiple instances of \OAK for each reported label type, it offers a more effective means of handling dependencies within the data. From a theoretical standpoint, \POAK deviates from the independence assumptions which underpin the AK principle theory presented in \citet{meir2023easy}.
Therefore to provide theoretical foundations for the \POAK algorithm, we establish a stronger version of the AK principle which encompasses per-reported-type estimations. From an empirical perspective, %we implement and test all three variants 
we show that the proposed algorithm substantially improves the cost-accuracy trade-offs compared with the baselines on several real-world datasets from various domains collected at Meta. We also propose a third variant, \POAKi, which reduces the number of latent variables by incorporating item response theory (IRT)~\cite{BakerBook2004item}.

%\jr{Can we provide some intuition about the proposed algorithm somewhere in the intro (or contributions)? Basically from this section, I'm only getting that it's an "extension" PTD. It'd be great to highlight what the overall algo and intuition in a sentence or two and say the difference with PTD.}\rmr{ok now?}
\newsubsec{Related Work}
\label{sec:related_work}

\newpar{Competence estimation} 
Given that crowdsourcing workers may possess vastly different capabilities due to differences in their inherent competence, training, or effort, it is crucial for crowdsourcing models to learn and account for worker accuracy in order to enhance ground truth estimation
~\cite{%ShengKDD2008,
IpeirotisDMKD2014,ZhengVLDB17}. 

In its simplest form, worker competence is captured by a single real number which represents the ability that each worker correctly answers a task~\cite{whitehill2009whose,karger2011iterative}. For categorical labels, one common way to characterize workers' performance is using a confusion matrix to model their ability to correctly label items of different categories~\cite{dawid1979maximum,RaykarJMLR2010,KimAISTATS2012}. Another line of work uses a multidimensional vector to model the diverse skill sets of each worker~\cite{WelinderNIPS2010,ZhouNIPS2012,MaKDD2015faitcrowd}. \rmr{perhaps fewer citations? \an{I've removed a bunch of long references.}}\rmr{I think two examples each time are enough :)}

%While algorithmic approaches and underlying statistical models vary, the main recurring assumption is that each worker can essentially be reduced to a single number representing her accuracy or competence~\cite{dawid1979maximum, karger2011iterative,li2014resolving}. Some algorithms add latent variables e.g. by modeling entire confusion matrices~\cite{?} or item difficulty~\cite{WhitehillNIPS2009}. Since these latent variables can often be estimated if we had the ground truth, and the ground truth can itself be estimated by properly aggregating labels, the common approach is to apply some sort of EM algorithm that iteratively estimates both.

Naturally, estimation in the above approaches employs statistical analysis that assumes specific label structure---typically multiple-choice questions or a real-valued number---and usually also a particular noise model. 

% Another approach is to use supervised learning, estimating workers' competence based on some `gold items' who are annotated by experts. However, such expert labels are of course much more expensive and are thus uncommon. 
\paragraph{Average similarity}
The idea of using average similarity of workers as a rough proxy for their competence had also been analyzed in specific domains and applied in practice in \emph{Games with a Purpose}~\cite{von2008designing}. Benefits of the average similarity approach were demonstrated theoretically and empirically in specific offline domains  including  abstractive-summarization models~\cite{kobayashi2018frustratingly} and binary labels~\cite{kurvers2019detect}.

\newpar{Complex annotations} As discussed above, the literature on aggregating complex annotations consists of many task-specific specialized models. \citet{NguyenACL2017} propose a HMM-based model for aggregating sequential crowd labels, which was applied to named-entity recognition and information extraction applications. \citet{linUAI2012crowdsourcing} introduce \textsc{LazySusan} to infer the correct answer of crowdsourcing tasks that can have a countably infinite number of possible answers. \citet{BransonCVPR2017lean} propose a model for several complex domains. %multi-object bounding box labels, single-object keypoint labels, and simple binary labels. 
Various other custom models were also proposed for different crowdsourcing applications such as co-reference~\cite{paunEMNLP2018probabilistic,liCOLING2020neural}, %pairwise ranking and preference~\cite{chenWSDM2013pairwise,LiWWW2022context}, 
sequence~\cite{rodriguesMLJ2014sequence}, and translation~\cite{zaidanACL2011crowdsourcing}.

Recent work leverages pairwise similarity as a general abstraction for complex annotations. The multidimensional annotation scaling (MAS) algorithm~\cite{braylan2020modeling} embeds the pairwise distances in a different space and estimates the competences that maximize the likelihood of observed distances. In a followup paper, the authors suggest a general way to aggregate complex labels~\cite{Braylan2021aggregating}. \citet{kawase2019graph} apply graph algorithms to find the `core' of the similarity graph, which is assumed to contain the most competent workers. 
The simplest approach proposed by \citet{meir2023easy} was shown to perform consistently well compared to the other approaches in an offline setting.

\newpar{Online crowdsourcing} The online crowdsourcing process that we focus on has been studied under different names in the literature including repeated labeling~\citep{DaiAI2013pomdp,IpeirotisDMKD2014,LinHCOMP2014relabel}, adaptive stopping~\cite{AbrahamSIGIR2016}, and incremental relabeling~\cite{DrutsaSIGMOD2020crowdsouringtutorial,DrutsaWSDM2020crowdsourcingtutorial}, in which various sophisticated algorithms were developed to guide the relabeling process. However, all of these methods, including the recent CLARA system developed at Meta~\cite{NguyenKDD2020}, only support simple categorical labels. To the best of our knowledge, our paper is the first to study a general approach for online crowdsourcing with general annotations.

%that we focus on in this paper has been studied under various different names.  incremental relabeling~\cite{DrutsaSIGMOD2020crowdsouringtutorial,DrutsaWSDM2020crowdsourcingtutorial}, repeated labeling~\cite{IpeirotisDMKD2014}, adaptive stopping~\cite{AbrahamSIGIR2016}
%relabeling:
%- Re-active Learning: Active Learning with %Relabeling~\cite{LinAAAI2016reactive}
%- To Re(label), or Not To Re(label)
%Authors~\cite{LinHCOMP2014relabel}
%- {POMDP}-based control of workflows for %crowdsourcing~\cite{DaiAI2013pomdp}

\newsec{Preliminaries}
\label{sec:prelim}

\newpar{Labels} A label is an element of some predefined set $\calX$, which can be either finite or infinite in nature. The most common problems are either categorical, where $\calX$ is some finite set of exclusive alternatives (e.g. male/female, types of fruits, names of authors, etc.); or real-valued, where $\calX$ are real numbers (representing temperature, price, etc). In this paper, `annotation' and `label' are used interchangeably.

\newpar{Similarity} The relation among different possible labels is captured by a \emph{similarity function} $s:\calX\times\calX \rightarrow \mathbb [0,1]$. For example, a commonly used similarity for categorical labels is the Hamming similarity, that is, $s(x,x')=1$ if $x=x'$ and otherwise $0$. 

\newsubsec{Classic Truth Discovery}
\label{sec:TD}
\newcommand{\hatm}[1]{\hat{\bm{#1}}}

In a classic truth discovery problem, input is given by a (possibly partial) $n \times m$ table $X$, where $x_{ij}\in \calX$ is the label reported for item $j\leq m$ by worker $i\leq n$. In addition, each item $j$ has a true answer $z^*_j\in \calX$. A \emph{classic truth discovery algorithm} is essentially a function  mapping input tables to a vector of predicted answers $\hatm z =(\hat z_j)_{j\leq m}$.

\newpar{Evaluation}
We evaluate the accuracy of an answer $z_j$ by its similarity to $z^*_j$. The accuracy of $\hatm z$ is simply the average over all items, i.e.:
% \begin{equation}
% \label{eq:ACC}
$ACC(\hatm z,\bm z^*) :=  \frac1m\sum_{j=1}^m s(\hat z_j,z^*_j)$.
%\end{equation}

\newpar{Average similarity and the AK principle}
\label{sec:PTD}
Given full input $X$,\footnote{We explain below the extension to partial data.} the average similarity of worker $i$ is 
\begin{equation}\label{eq:pi_full}
    \pi_i := \frac{1}{n-1}\sum_{i'\neq i} \frac1m \sum_{j\leq m} s(x_{ij}, x_{i'j}).\vspace{-0mm}
\end{equation}
%\jr{Should there be two sums in the above equation? One for workers and one for items?}
The average similarity can be computed directly from the input data. In contrast, we often assume that each worker has some intrinsic, unobserved competence that determines her accuracy. 
We define the competence as $c_i:=E[s( x_i, z^*)]$, but note that for this to be well-defined, we need a specific noise model---a distribution connecting the ground truth $z^*$ with the worker's label $x_i$. 

For example, a common noise model for binary labels is the \emph{one-coin} model (also known as the binary  Dawid-Skene~\shortcite{dawid1979maximum} model), where $x_{ij}=z^*_j$ with some fixed probability $p_i$, independently for every item. Note that under the one-coin model and Hamming similarity, the competence of each worker $i$ is exactly $c_i=p_i$.  For the one-coin model, a linear connection between $c_i$ and $E[\pi_i]$ was shown by \citet{kurvers2019detect}. In \citet{meir2023easy}, it is shown that this linearity holds for any label type, exactly or approximately, under a wide range of noise models as per the Anna Karenina principle. Their \PTD algorithm is essentially:
\begin{enumerate}
    \item Calculate $\pi_i$ for each worker;
    \item Apply a linear transformation to get $\hat c_i$ from $\pi_i$;
    \item Get each $\hat z_j$ by a weighted aggregation of the labels $(x_{ij})_{i\leq n}$, where weights depend on $(\hat c_i)_{i\leq n}$.
\end{enumerate}

\full{The default implementation of the algorithm simply uses $\hat c_i = \pi_i$ as weight. However the authors also suggest other implementations (i.e. with specific transformations in steps 2 and 3) that are optimal under particular noise models. }

\begin{table}
\begin{tabular}{cl}
    \textbf{Notation} & \textbf{Description} \\
    \hline
    $N, M$ & Set of unique workers / items \\
    $x_{ij}$ & Observed label by worker $i$ to item $j$  \\
    $z^*_j, \hat z_j $ & True / estimated annotation for item $j$ \\
    $c_i, \hat c_i$ & True / estimated  competence of worker $i$\\
    $\hat c^0_i$ & Competence estimated from true labels\\
    $M_i,M^*_i$ & Items [with ground truth] labelled by $i$\\
    $M_{ii'}$ & Items labeled by both $i$ and $i'$\\
    $\overline m_i$ & Number of pairwise comparisons involving $i$\\
    $m_i,m^*_i, m_{ii'}$ & Size of $M_i,M^*_i, M_{ii'}$ \\
    $\pi_i$ & Average similarity of $i$ to other workers\\
    \hline
\end{tabular}
\caption{Common notations used in the paper. Lower-case letters $n,m$ represent set sizes, e.g. $m^*_i=|M^*_i|$.}\label{tab:notation}
\end{table}

\newpar{Partial data}
In general, every item may be labelled by a subset of workers. We denote by $N_j :=\{i\in N: x_{ij} \text{ exists}\}$ the set of workers labelling item $j$. Similarly, $M_i :=\{j\in M: x_{ij} \text{ exists}\}$ is the set of items labelled by worker $i$. We also denote $M_{ii'}:=M_i \cap M_{i'}$. Note that Eq.~\eqref{eq:pi_full} only applies for full data. For a general partial matrix, we compute $\pi_i$  by taking the average similarity over every label in every $M_{ii'}$. %More details in the description of the online algorithm.

%\jr{can we move the last sentence or this section earlier (or mention again)? I actually was wondering what we'd do with partial data when looking at equation 2 initially.}\rmr{done}

\newsubsec{Online Crowdsourcing}
\label{sec:DMR}
In an online setting, there is no predefined set of items and workers. Instead, there is a dynamic pool of workers, and items arrive sequentially. Upon arrival of an item $j$, we can ask for a label $x_{ij}$ from worker $i$. In this paper, we assume the decision on which worker to ask is made by a separate system that we have no control over. %\footnote{For example, we cannot decide to only ask for labels from competent workers.} 
Therefore, an online truth-discovery system  makes the following  decisions for each item: (1) aggregate collected labels, (2) estimate the accuracy of individual/aggregated labels, and (3) decide whether we should stop labeling. See Fig.~\ref{fig:online_general} for a diagram of the labeling process. 

%\begin{enumerate}
%    \item Aggregate collected labels;
%    \item Estimate the accuracy of individual/aggregated labels;
%    \item Decide when we have enough labels.
%\end{enumerate}

%We note that the third decision on where to set the threshold is a managerial rather than algorithmic decision. \jr{this sentence is confusing. Probably it's good to just say that in this paper we're focusing on whether to re-review after the first review since most practical applications have at most 3 labels. } 
Due to constraints on labeling capacity, in most practical tasks the number of collected labels per item is at most 3, and thus aggregation is rather straightforward. Therefore, our goal in this paper is to obtain a good estimation of the label accuracy. Moreover, we will focus on the first decision point as it is most crucial in our crowdsourcing tasks.

\newpar{Cost-quality trade-off}
The main implication is that the task reduces to estimating label quality given  available information. The performance  is then measured by the system cost-quality trade-off: how many labels per item are needed to reach a certain level of quality. An example trade-off curve for the \POAK algorithm is presented in Fig.~\ref{fig:dmr_example}. The selection of threshold is a case-by-case decision that hinges on labeling costs and the acceptable margin of error~\cite{NguyenKDD2020}.

% \begin{figure}[t]
%     \centering
%     \includegraphics[width=\linewidth]{figures/binary_example_curve.png}
%     \caption{An example of the cost-accuracy trade-off curve on synthetic data with 7 binary (yes/no) labels per item.  The Baseline shows the accuracy obtained by taking the majority of the first $1, 2, 3, \cdots$ labels. The thin blue line shows the accuracy of the  algorithm as we increase the threshold. \jr{is the baseline using majority vote; if so, mention that?} \rmr{remove for space? \an{Agreed. We should remove this and/or combine with Fig. 3 below if needed.}}}
%     \label{fig:curve_example}
% \end{figure}

\begin{figure}
    \centering
    \includegraphics[width=0.8\linewidth]{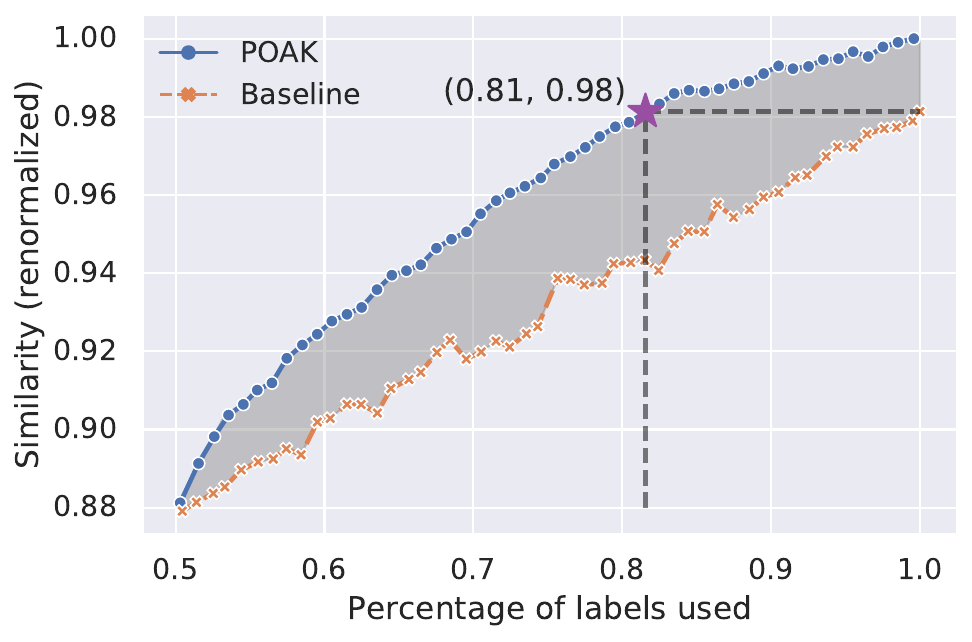}
    % \vspace{-2mm}
    \caption{
    Performance of the \POAK algorithm on the Keypoints dataset, compared to a baseline that decides randomly on how many labels to use. Each point on the curve corresponds to the percentage of labels used and associated similarity at a given accuracy threshold. The star marker indicates that \POAK achieves on-par accuracy with the baseline only using 81\% of the labels. The shaded area (relative AUC) measures the improvement over the baseline.}
    \vspace{-0mm}
    \label{fig:dmr_example}
\end{figure}

% As a very simple example of such a process that does not require any model, we can consider a system that  estimates  the accuracy of $\hat z_j$ after $i\geq 2$ labels as $\hat C_j=\frac{|\{i':x_{i'j}=\hat{z}_j\}|+1}{i+2}$. Then, given a threshold $T$, the system will require the $i+1$st label only if $\hat C_j<T$. 

%More generally, the decision on when to stop collecting labels may be complicated, and rely on statistics collected on previous items and the  margin of error we can tolerate. For example, Facebook is deploying a system that estimates confidence intervals for individual and aggregated answers in order to efficiently label policy violations in user posts~\cite{NguyenKDD2020,NguyenKDD2022}. \rmr{do you want to add here?}

\newpar{Auditor labels}
Some items may arrive with an `auditor label', that we can think of as the ground truth. The decision on whether to ask for an auditor label on a particular item is done independently by a different system, and we assume these labeled items are a random sample from all items. We denote by $M$ the set of all items, and by $M^* \subseteq M$ the items for which an auditor label is available. The auditor label is denoted by $z^*_j$ (i.e. it is considered to be the ground truth).  

 % An online algorithm may use all data collected up to certain time as a training set, later using the learned patterns to predict whether the current answer is sufficiently likely to be correct. An intermediate evaluation metric for such algorithms is therefore how well-calibrated they are. That is, whether their prediction about the accuracy of the current label correct is not too optimistic, neither too pessimistic, in expectation. 
%%%%%%%%%%%%%%%%%%%%%%%%%%%%%%%%%%%%%%%%%%%%%%%

%As mentioned above, we want to adapt the algorithm from \cite{meir2023easy} that is based on average similarity, in order to apply it in an online setting.
\newpar{Online algorithms}
An online accuracy estimation algorithm is composed of two components (see Fig.~\ref{fig:online_general}):

%The general solution is also explained above, and here we provide some more detail on how we adapt PTD to PTD-DMR, and explain the nuances. The adaptation of MAS is similar.

\begin{description}
    \item \textbf{A learning component} that gets as its input a partial matrix $X$, possibly with some auditor labels, and outputs a model $\Theta$. 
    \item \textbf{An estimation component} that gets as input a set of reported labels for a particular item $(x_{ij})_{i\in N_j}$ and has access to the model $\Theta$. It outputs an aggregated answer $\hat z_j$ with its estimated accuracy $\hat C_j$.
\end{description}
\newsec{Online Anna Karenina Algorithm}
\label{sec:average}

We start with a simple model that only includes the estimated competence of each worker. 
Note that a straight-forward way to estimate the competence, is to consider the average accuracy over items with auditor label:
$\hat c^0_i := \frac{1}{| M^*_i|}\sum_{j \in M^*_i}s(x_{ij}, z^*_j)$. Clearly, if $M^*_i$ is sampled randomly from $M$, then $\hat c^0_i$ is an unbiased estimator of $c_i$. However, typically auditor labels are expensive and hence $M^*_i$ is small or empty. %Therefore relying on supervised learning alone would be very inefficient.

\newsubsec{The Learning Component}
\label{sec:learning}

The first step is to calculate $\pi_i$ for each worker. %However since $X$ is only a partial matrix, Eq.~\eqref{eq:pi_full} cannot be applied directly. %By the Anna Karenina principle~\cite{meir2023easy}, this already provides some crude approximation of the true competence $c_i$. Our next steps refine this estimation. 
%\newpar{Partial data and regularization}
Note that under a full input matrix, all workers are treated the same as in  Eq.~\eqref{eq:pi_full}. However in practice a worker may have co-labelled many items with some workers, and just a few with others, resulting in noisy pairwise similarity with the latter group. We therefore compute 
$$\pi_i:=\frac{1}{\overline m_i}\sum_{i'\neq i}\sum_{j\in M_{i}}s(x_{ij},x_{i'j})$$ 
where $\overline m_i:=\sum_{i'\neq i}m_{ii'}$, that is, we average over all pairs of $i$'s label and another label of the same item.

\rmr{I removed the part on weight T() since it is not what we implemented anyway (I wanted to but did not have time)}

% The other extreme approach would be to compute $\pi_i$ by summing over all pairwise labels, regardless of who is the other worker. This carries the risk of over-weighting a small subset of co-workers with large intersection, but  that might not be representative. 

% An intermediate solution is to compute the pairwise similarity $s_{ii'}$ with each worker independently, but take the number of samples $m_{ii'}:=|M_{ii'}|$ into account when calculating $\pi_i$. Specifically, we decide how much weight we attribute to small samples by deciding on a weight function $T$ (increasing, concave),  and set $\pi_i := \sum_{i'\neq i} T(m_{ii'}) s_{ii'}$\jr{Would be great to define what $T(m_{ii'})$ represents specifically.}. E.g. when $T$ is the identity function all samples are equal, and when $T(m_{ii'}) = \min\{1,m_{ii'}\}$ then all co-workers are equal.\footnote{These two options are optimal }
% The total number of samples $m_i=\sum_{i'\neq i}T(m_{ii'})$ is also recorded as part of the model to allow for regularization later in the estimation component.

\newpar{Calibration and semi-supervised learning}
While $\pi$ and $c$ are positively correlated, there may be a better linear transformation between them than simply using identity. Indeed, if we have some additional statistical assumptions on the data we can analytically derive such a transformation~\cite{kurvers2019detect,meir2023easy}, but having access to a small amount of audited labels allows us to take an easier and more general approach. 

Recall that we have our preliminary competence estimation $\hat c^0$ that is based on the supervised data. While on their own they may be too noisy, we can use them to calibrate the competence estimation, by computing the best linear transformation $L$ from $\pi= (\pi_i)_{i\in N}$ to $(\hat c^0_i)_{i\in N}$. This transformation has only two latent variables (slope and intercept) so even a small amount of supervised data is sufficient. % for accurate estimation. 
If auditor labels are available, we can also combine $\hat c^0_i$ and $\pi_i$ to get a better estimate of $c_i$. While this aspect is crucial in practice, it is also relatively straightforward, so we defer the details to the Appendix%App.~\ref{apx:supervised}
. We summarize the steps of the learning component of our \OAK algorithm in Alg.~\ref{alg:OAK}.

\full{For this we use weighted linear regression, where the weight of every `sample' $(\pi_i,\hat c^0_i)$ is $|M^*_i|$. 
}

\par
\begin{algorithm}[t]
\caption{\OAK\textsc{(Learning Component)}}
\label{alg:OAK}
\KwIn{dataset $X$}
\KwOut{model $\Theta$} 
\For{$i\in N$} {
% \For{$i' \neq i$} {
%     Set $m_{ii'} := |M_{ii'}|$\;
    
%     Calculate $s_{ii'} := \frac{1}{m_{ii'}} \sum_{j\in  M_{ii'}} s(x_{ij}, x_{i'j})$\; 
% }
Set  $\overline m_i := \sum_{i'\neq i}|M_{ii'}|$\;

Set $\pi_i:=\frac{1}{\overline m_i}\sum_{i'\neq i}\sum_{j\in M_{i}}s(x_{ij},x_{i'j})$\;

Set $\hat c^0_i := \frac{1}{|M^*_i|}\sum_{j \in M^*_i}s(x_{ij}, z^*_j)$\;
}
Find the best linear transformation $L$ from $\pi$ to $c^0$ using weighted linear regression\;

%For each $i\in N$, set $\hat c_i := \frac{\alpha^{} m^*_i}{\alpha^{} m^*_i+ m_i}c^0_i +  \frac{m_i}{\alpha^{} m^*_i+ m_i} L(\pi_i)$\;
For each $i\in N$, set $\hat c_i :=  L(\pi_i)$\;

\Return $\Theta := (\hat c_i, \overline m_i)_{i\in N}$\; 
\end{algorithm}

\rmr{Maybe we should remove all references to auditor labels and calibration from the main text? it is not required, and the effect in practice is small. We can keep it in the appendix. }
\newsubsec{The Estimation Component}
Estimation in \OAK is rather straightforward. 
Suppose that we have a new item $j$ with labels $(x_{ij})_{i\in N_j}$. First, the algorithm calculates the current estimated label by aggregating all current labels $\hat z_j := agg((x_{ij})_{i\in N_j})$. Recall that the aggregation function is decided up front. Then the algorithm predicts the accuracy of $\hat z_j$ (denoted $\hat C_j$) using the model $\Theta$. This  estimation is then compared to a pre-defined threshold; see Fig.~\ref{fig:online_general}. % \jr{Or threshold can be selected by plotting a cost-quality plot, e.g., similar to how thresholds are selected for ML models through an ROC or Precision-Recall curve.}  \rmr{this is not something we can do in production. it is part of the evaluation.}

\newpar{Estimating accuracy of a single label}
The most important part is to estimate the accuracy of the first label, since asking for another label will at least double the cost. Here aggregation is trivial, as $\hat z_j = x_{ij}$ for the first worker $i\in N_j$.

The easiest estimation is just by setting $\hat C_j := \hat c_i$, i.e. relying completely on the estimated competence of the reporting worker.
To deal with small samples we apply additive smoothing, which shrinks the estimation towards the mean accuracy of the entire population, defined as $\bar c:=\frac{1}{m}\sum_{i\in N_j}m_i \hat c_i$, see Alg.~\ref{alg:OAK_predict}. The $\gamma$ is a meta-parameter, set to 10 by default.

% This may work reasonably well when $\hat c_i$ was estimated based on a large sample, but many workers in the data may have few samples or none at all. Here regularization kicks in, as we can use $m_i$ to know how reliable $\hat c_i$ is. In the extreme case where $m_i=0$ (or below some threshold), we use the baseline accuracy 
% %$c^0 := \frac{\sum_{j\in M^*}\sum_{i\in N_j}s(x_{ij},x^*_j)}{|\{\sum_{j\in M^*}\sum_{i\in N_j}\}|}$.
% $\bar c:=\frac{1}{m}\sum_{i\in N_j}m_i \hat c_i$. 
% More generally, we set 
% $$\hat c(x_{ij}) := F(m_i)\hat c_i + (1-F(m_i))\bar c,$$
% where $F$ is some increasing linear function.

\newpar{Aggregated labels}
Estimating the accuracy of an aggregated label is tricky. % Even if the aggregation is simple (say, take the annotation of the most competent worker),
A naive approach that returns the estimated accuracy of one worker only is missing important information: if the other reported labels are identical or similar, this is a strong positive signal, whereas if we know that other workers reported different labels this may suggest lower accuracy. The solution is detailed in the Appendix.
%App.~\ref{apx:acc_aggregated}. %In the simple variant we ignore this issue, and just estimate accuracy based on the worker closest to $\hat z_j$. 

% \medskip
% As with the learning component, we summarize the steps of the estimation component in Alg.~\ref{alg:OAK_predict}. 

\begin{algorithm}[t]
\caption{\textsc{\OAK (estimation Component)}}
\label{alg:OAK_predict}
\KwIn{Labels $(x_{ij})_{i\in N_j}$, model $\Theta$}
\KwOut{Estimated label $\hat z_j$, estimated accuracy $\hat C_j$} 

Aggregate labels $\hat z_j := agg((x_{ij})_{i\in N_j})$\;

Find closest worker $i^* := \arg \max_{i\in N_j}s(x_{ij},\hat z_j)$\;

Calculate $\hat C_j := \frac{\overline m_{i^*}}{\overline m_{i^*}+\gamma}\hat c_{i^*} + \frac{\gamma}{\overline m_{i^*}+\gamma}\bar c $\;

\Return $(\hat z_j, \hat C_j)$\;
\end{algorithm}

\begin{algorithm}[t]
\caption{\textsc{\POAK (estimation Component)}}
\label{alg:POAK_predict}
\KwIn{Labels $(x_{ij})_{i\in I_j}$, model $\Theta= (\Theta^{(\ell)})_{\ell\leq k}$}
\KwOut{Estimated label $\hat z_j$, estimated accuracy $\hat C_j$} 

Aggregate labels $\hat z_j := agg((x_{ij})_{i\in I_j})$\;

Set $\ell$ such that $\hat z_j \in \calX^{(\ell)}$\;

Find closest worker $i^* := \argmax_{i\in N_j}s(x_{ij},\hat z_j)$\;

Calculate $\hat C_j := \frac{\overline m^{(\ell)}_{i^*}}{\overline m^{(\ell)}_{i^*}+\gamma}\hat c^{(\ell)}_{i^*} + \frac{\gamma}{\overline m^{(\ell)}_{i^*}+\gamma}\hat c_i$\;

\Return $(\hat z_j, \hat C_j)$\;
\end{algorithm}

\newsubsec{Algorithm Complexity}
On time complexity, we need to calculate pairwise similarity between labels within an item for all items. Given that each item contains only a few labels, the time complexity is linear in number of items $\mathcal{O}(M)$, or total number of labels $\mathcal{O}(\sum_{i}\sum_{i'}|M_{ii'}|)$. Regarding memory complexity, we need to store the estimated confidence of each labeler per annotation type which is $\mathcal{O}(kN)$. As both the time and memory complexities are linear, we believe that our algorithm can be implemented efficiently with little resource concern.
%\rmr{\newsubsec{Empirical Results}not sure if we want to already put some results here. I think so since there are no theoretical results in this part}

%%%%%%%%%%%%%%%%%%%%%%%%%%%%%%%%%

\newsec{Annotation Types}
To demonstrate the main issue we tackle in this paper, consider a population with two types of workers and two equally-frequent labels: type~$a$ always identify label $A$ correctly, and type~$b$ always identify label~$B$ correctly. Each worker makes mistakes on the other label with probability $0.5$. Assuming equal priors, each worker has an overall accuracy of $c_i=0.75$, and every worker has an expected average similarity of $E[\pi_i]= (0.75+0.5)/2 = 0.625$.\footnote{Since in case the label matches the type agreement is $1$ with her own type and $0.5$ with the other type (overall $0.75$) and if label mismatches type then agreement is $0.5$ with any other worker.}  In particular \OAK is not able to distinguish between workers since both are equally competent, and can only assess the accuracy of the first label as $0.625$, regardless of the label or the identity of the worker. However if the reported label is $B$ and we know that the worker is of type~$a$,  then we could tell for sure that the label is correct. In contrast, if the worker is of type~$b$ then we know the expected accuracy is only $2/3$.

A different reason that can cause a similar problem is label frequency. Suppose all workers are correct with probability $c_i=2/3$ regardless of the label, and label $A$ is five times more frequent than $B$. Clearly, a reported label $B$ is much less likely to be correct than a reported label $A$ (the posteriors are $2/7$ vs. $10/11$ respectively). Yet the simple \OAK algorithm will predict the same accuracy in both cases.  

\newsubsec{Partition-based \OAK Algorithm}
%The common approach to deal with issues as the one above, is to explicitly learn the parameters of the (assumed) noise model, as well as the priors. Since this approach would be very domain-specific, 

We propose a general approach to deal with the above issues without explicitly assuming the underlying model or priors. Our approach is based on a conditional application of the AK principle. 
We will first describe the modified algorithm, and then explain the theoretical justification. 
In short, we first partition the space of labels into $k$ types $\calX = \uplus_{\ell=1}^k \calX^{(\ell)}$.

The partition itself is decided externally using domain specific knowledge, where the guiding principle is that  similar labels should be grouped together. For example, if there are few categorical labels then each category is a separate type; and if annotations are free-form sentences, the type can be determined by some syntactic feature of the label or some classification of the words within. Alternatively, it is feasible to employ clustering algorithms to automatically determine label types.

\full{It is also possible to automatically determine label types using some clustering algorithm---this whole process is transparent for our model.  }

Building upon this partitioning idea, instead of assigning a single latent variable per worker to represent their accuracy, we assign $k$ latent variables for each worker, where $c_i^{(\ell)}:=E[s(x_i,z^*)| x_i \in \calX^{(\ell)}]$ is $i$'s conditional accuracy for label $\ell$. Note that we condition on the reported label rather than the true label. \full{In the special case of categorical labels, $c_i^{(\ell)}$ is exactly the posterior probability that a reported label $\ell$ is correct, when reported by worker~$i$.}

\full{
\newpar{A larger model}
For categorical data (say, with $k$ categories) the discussion above suggests a rather straightforward solution. We can extend the model so as to learn both prior probabilities and full $k\times k$ confusion matrices for each worker. However such a drastic complication of the model is both infeasible (unless we have a huge amount of data for every worker), and not generalizable to complex labels.
}

The learning component of the \POAK algorithm applies Alg.~\ref{alg:OAK} on each $\calX^{(\ell)}$ separately, in order to get a conditional average similarity $\pi^{(\ell)}_i$ and conditional accuracy estimate $\hat c^{(\ell)}_i=L^{(\ell)}(\pi_i^{(\ell)})$ for every worker $i\in N$ and label type $\ell\leq k$. The new, larger model will then contain $(\hat c_i^{(\ell)},\overline m_i^{(\ell)})_{i\in N, \ell\leq k}$. Note that there is no explicit estimation of labels' priors or noise model parameters. In its estimation component, the \POAK algorithm picks the type $\ell$ of the aggregated label, and sets the estimated accuracy $\hat C_j$ according to $\hat c_i^{(\ell)}$, with additive smoothing towards $\hat c_i$. See Alg.~\ref{alg:POAK_predict}.

%  \jr{As discussed in the meeting, would be great to expand on this and have guiding principles for how to select the annotation types supported through numerical results. This seems like a key part of the novelty in the complex annotations setting.}\rmr{I am not sure I agree. Ideally, we would show that performance is not very sensitive to the way we partition, and essentially any `reasonable' partition improves performance.}

% \newsubsec{Conditional Similarities}
% The main challenge we face it to estimate the conditional accuracies. Using only auditor labels for supervised learning we can easily calculate $c^{(\ell)}_i$ for every worker and label type. However this estimation is even more sparse than in the simple algorithm, so the question is how to use the pairwise similarity information.

% The answer lies in using \emph{conditional similarities}. Formally, we define all of $s^{(\ell)}_{ii'}, m^{(\ell)}_{ii'}$ and $\pi^{(\ell)}_i$ similarly to the type-oblivious definitions above, except that we only consider samples where $x_i \in \calX^{(\ell)}$.  Then, we want to use the Anna Karenina principle to argue that the conditional competence vector $(c^{(\ell)}_i)_{i\in I}$ is linear in (and thus can be easily estimated from)  the conditional similarity vector $(\pi_i^{(\ell)})_{i\in I}$, for every label $\ell$. 

\newsubsec{Theoretical Justification}
To justify the \POAK algorithm, we need to show that $\pi_i^{(\ell)}$ is linear in $c_i^{(\ell)}$ in expectation. \full{In other words, that the AK principle applies even when conditioning on the reported label (or more generally on its type). } This is not \textit{a priori} clear, as one of the assumptions in \citet{meir2023easy} is  that labels from workers are independent conditional on their accuracy---an assumption that is violated once conditioning on the reported label.  

\full{
Now, it can be shown that the proof of the main theorem in \cite{meir2023easy} still goes through if we condition on $x_{ij}$. However this would neither provide us with an accurate linear relation, nor it would work for the other theorems, as they make some explicit assumptions on the distribution of labels. }

We therefore establish an alternative version of the Anna Karenina principle for conditional categorical labels. In particular we provide sufficient conditions to an exact linear relation between $c_i^{(\ell)}$ and $E[\pi_i^{(\ell)}]$, which is positive if (but not only if!) the overall population accuracy is better than random guess. 

% In particular our condition holds either if there are only two labels (the binary case), or in the case of equal priors. We note that our result still allows for workers with arbitrary types (confusion matrices), whereas the result in \cite{meir2023easy} requires a very specific structure.   

\newpar{Statistical model}
We adopt the general Dawid-Skene model for categorical labels~\cite{dawid1979maximum}:
\begin{itemize}
    \item Prior probabilities over labels, denoted by $q^{(\ell)}$ with $\sum_{\ell\leq k} q^{(\ell)}=1$;
    \item Worker types. A type is specified by a $k\times k$ confusion matrix $\mat_i$, where $\mat^{\tau\rightarrow \ell}_i:=Pr[x_i =  \ell| z^* = \tau]$. We require $\sum_{\ell\leq k}\mat_i^{\tau\rightarrow \ell}=1$ for all $i,\tau$; 
    \item Conditional independence among workers: $Pr[x_i = \ell | z^* = \tau, x_{i'}]=\mat^{\tau\rightarrow \ell}_i$.
\end{itemize}
\citet{meir2023easy} imposes a strong restriction on the model where each $\mat_i$ depends only on the accuracy of $i$, which is a scalar parameter, in order to derive the corresponding linear relation. %, and proved a linear relation between the competence of $i$ and the expected similarity to $i'$ (extending a similar result on binary labels from \cite{kurvers2019detect}). 
We make no assumption on the confusion matrices, and ask how expected conditional similarity behaves as function of the conditional accuracy $c_i^{(\ell)}$. 

Denote by $p^{(\ell)}_i:=\mat_i^{\ell\rightarrow\ell}$ the probability of $i$ correctly identifying label $\ell$. We define a \emph{partial type} $\mat_i^{(-\ell)}$ which is $\mat_i$ without the column $\mat_i^{\ell\rightarrow \cdot}$. Then we can fix the partial type, and ask how both $\pi^{(\ell)}_i$ and $c_i^{(\ell)}$ change as a function of the remaining latent variables, and in particular $p_i^{(\ell)}$. 
Our main result is the following.
\begin{theorem}[Conditional Anna Karenina theorem for categorical data]
Fix prior probabilities $q$, a category $\ell$, and a worker $i$ with partial type $\mat^{(-\ell)}_i$. Then there are constants $\alpha^{(\ell)}, \beta^{(\ell)}$ such that
$E[\pi_i^{(\ell)}]=\alpha^{(\ell)} c_i^{(\ell)} + \beta^{(\ell)}$.
\end{theorem}
\full{Note that the expectation is taken over both the types of the other workers, and their reported labels.}

We present the key components of the proof here and defer the complete proof to the Appendix. %App.~\ref{apx:POAK}. 
First, we note that since $\pi_i^{(\ell)}$ is an average over comparisons to random labels reported by a random worker $i'$, we have (omitting the item subscript $j$):
\begin{small}
$$E[\pi_i^{(\ell)}]=E_{i'}[E[s(x_{i'},x_i)|x_i=\ell]]=E_{i'}[Pr[x_{i'}=\ell|x_i=\ell]].$$
\end{small}
Hence by linearity of expectation, it is sufficient to show the following proposition. % that $Pr[x_{i'}=\ell|x_i=\ell]$ is linear in $c_i^{(\ell)}$ for any  type $\mat_{i'}$.

\begin{proposition}$Pr[x_{i'}=\ell|x_i=\ell]$ is linear in $c_i^{(\ell)}$, for any worker type $\mat_{i'}$.
\end{proposition}
\begin{proof} We split to cases when $x_i$ agrees or disagrees with the truth $z^*$, and show that both terms are linear in $c_i^{(\ell)}$:
\begin{small}
\begin{align}
  Pr&[x_{i'}=\ell|x_i=\ell] = 
        Pr[x_{i'}=\ell|z^*=\ell] c_i^{(\ell)} \\
        &~~~+ \sum_{\tau\neq \ell}Pr[x_{i'}=\ell|z^*=\tau]Pr[z^*=\tau|x_i=\ell]\notag\\
                   &= p_{i'}^{(\ell)} c_i^{(\ell)} + \sum_{\tau\neq \ell}M_{i'}^{\tau\rightarrow \ell}\frac{q^{(\tau)} M_i^{\tau\rightarrow \ell}}{Pr[x_i=\ell]} \\
                   &= p^{(\ell)}_{i'} c_i^{(\ell)} + \frac1{Pr[x_i=\ell]}\sum_{\tau\neq \ell}q^{(\tau)} M_i^{\tau\rightarrow \ell} M_{i'}^{\tau\rightarrow \ell} \label{eq:sbb}
\end{align}
\end{small}
The first term in Eq.~\eqref{eq:sbb} is obviously linear in $c_i^{(\ell)}$ since $\mat_{i'}$ is fixed. For the second term, we first observe that each $q^{(\tau)} M_i^{\tau\rightarrow \ell} M_{i'}^{\tau\rightarrow \ell}$ for $\tau\neq \ell$ is completely determined by the fixed terms $q,\mat_{i'}$ and $\mat^{(-\ell)}_i$. It remains to show that $\frac1{Pr[x_i=\ell]}$ is linear in $c_i^{(\ell)}$. Note that
\begin{small}
\begin{align}
      c_i^{(\ell)}  &= \frac{q^{(\ell)} p_i^{(\ell)}}{Pr[x_i=\ell]}\label{eq:cil}\\
    Pr[x_i=\ell]  &= q^{(\ell)} p_i^{(\ell)}+ \sum_{\tau\neq \ell}q^{(\tau)} \mat_i^{\tau\rightarrow \ell} \label{eq:pxil}\\
\forall Y,Z>0,&~~\frac{1}{Z+Y} = \frac{-1}{Y} \cdot \frac{Z}{Z+Y} + \frac{1}{Y} \label{eq:YZ}
\end{align}
\end{small}
Finally, denote $Z:= q^{(\ell)}p_i^{(\ell)}$; $Y= \sum_{\tau\neq \ell}q^{(\tau)}\mat_i^{\tau\rightarrow\ell}$, and note that $Y$ is a constant. We have
$$\frac1{Pr[x_i=\ell]} = \frac{1}{Z+Y} =  \frac{-1}{Y} \cdot \frac{Z}{Z+Y} + \frac{1}{Y} =  \frac{-1}{Y} \cdot c_i^{(\ell)} + \frac{1}{Y},$$
as required.\end{proof}
%A sufficient condition for the linear relation to be positive, is that the overall probability to report $\ell$ correctly is higher than the probability to falsely report $\ell$, for any true label $\tau$.
%\rmr{In practice the relation seems to be very close to linear. One theoretical result we can show is ignoring the dependency as if each conditional relation is still linear. Another thing we can do is show empirically for randomly generated types that correlation is close to 1. }

\begin{remark}
The higher $\alpha^{(\ell)}$ is w.r.t $\beta^{(\ell)}$, the better we can separate good workers from poor ones. Note that a sufficient condition for $\alpha^{(\ell)}$ to be positive, is that $E_{\mat_{i'}}[p_{i'}^{(\ell)}]$ is higher than $E_{\mat_{i'}}[\mat_{i'}^{\tau\rightarrow \ell}]$ for all $\tau\neq \ell$. I.e., that the overall competence to identify $\ell$ is higher than the overall chance to incorrectly report a label as $\ell$.  \rmr{add proof?}
\end{remark}

\newsubsec{Item Response Theory}
Item Response Theory is a statistical model developed for standardized tests, which posits that the probability of answering a question correctly as a function of three factors of the question: 

\def\Pirt{\textsc{P}^{IRT}}
 \begin{equation*}
     \label{eq:IRT_main}\Pirt(x_{ij})=p^0_j+ \frac{1-p^0_j}{1+\exp(b_j(d_j-c_i))},
 \end{equation*}
 where $d_j, b_j, p^0_j$ correspond to item~$j$'s difficulty, separation, and base rate; and $\Pirt$ is the probability that answer $x_{ij}$ is correct~\cite{BakerBook2004item}. 
 We adopt the concept and the formula but use it in two unusual ways: 
\begin{itemize}
\item We assign the three `question factors' mentioned above not to each item, but to each \emph{reported label} type $\ell$. 
\item We do not restrict its usage to categorical data.
\end{itemize}
This means we define a model with $(3k+n)$ latent variables, compared to $(n\times k)$ for \POAK, that captures the expected similarity of $x_{ij}$ to the ground truth, conditional on $x_{ij}\in\calX^{(\ell)}$.

In the Appendix%App.~\ref{apx:IRT}
, we describe the \POAK-IRT (or \POAKi) algorithm, which first runs \POAK and then reduce the number of parameters by finding the IRT parameters with the best fit.

We also prove that the common case of single-parameter workers with arbitrary priors on true answers is captured by our IRT model without any loss of precision. 
\newsec{Empirical Results}
\label{sec:results}
To showcase the effectiveness of our proposed methods, we conduct extensive numerical experiments on four complex annotation datasets and compare against several baselines. 

\def\AGE{\textsc{Age}\xspace}
\def\KEYPOINT{\textsc{Keypoint}\xspace}
\def\TOPIC{\textsc{Taxonomy}\xspace}
\def\TREEPATH{\textsc{TreePath}\xspace}
\def\BOUNDINGBOX{\textsc{BoundingBox}\xspace}

\begin{table*}
\begin{minipage}[t]{.45\textwidth}

\centering
\begin{small}
   \begin{tabular}{lcl}
       \hline
       \textbf{Method} & \textbf{Conf.-based} & \textbf{Aggregation}\\
       \hline
       \POAK-Weight & Yes & Averaging/voting \\
       \POAK-BAU & Yes & Selection \\
       \OAK-Weight & Yes & Averaging/voting \\
       \OAK-BAU & Yes & Selection \\
        \SAD & No & Selection \\
       \Uniform & No & Averaging/voting \\
       \hline
   \end{tabular}
\end{small}
\caption{Summary of key characteristics of competing methods in our experiments.}
\label{tab:mathods}
\end{minipage}%
\hfill
\begin{minipage}[t]{0.49\textwidth}
\centering
\begin{small}
   \begin{tabular}{llllc}
       \hline
       \multirow{2}{*}{\textbf{Dataset}} & \multirow{2}{*}{\textbf{Labelers}} & \textbf{Similarity} & \textbf{Audit} \\
       & & \textbf{measure} & \textbf{labels} \\
       \hline
       \KEYPOINT & $\sim50$ & Gaussian similarity & Yes \\
       \TOPIC & $\sim500$ & Jaccard & Yes \\
       \TREEPATH & $\sim70$ &Hierarchical Hamming &   No \\
       \BOUNDINGBOX & $\sim100$ & Image Jaccard & Yes\\
       \hline
   \end{tabular}
\end{small}
\caption{Basic information of the real crowdsourcing datasets.}
\label{tab:datasets}
\end{minipage}\vspace{-0mm}
\end{table*}

% To showcase the effectiveness of the proposed method on labeling cost-quality trade-off, we conduct extensive numerical experiments on four complex annotation datasets and compare against several baseline methods. The results suggest that our method consistently outperforms the competing approaches under various scenarios.

\newsubsec{Setup of Numerical Experiments}
\newpar{Baselines and datasets}
To evaluate the three variants of our proposed algorithms (\OAK, \POAK, and \POAKi), we consider several baseline methods and label aggregation methods (summarized in Tab. \ref{tab:mathods}) including :
\begin{itemize}
    \item \SAD: Smallest Average Distance \citep{braylan2020modeling}. For a job with multiple labels, select the label which has the smallest average distance to other labels. 
    \item \BAU: Best Available User \citep{braylan2020modeling}. Given the (estimated) confidences of multiple labels for a job, select the label with the highest confidence. This aggregation method can be combined with any confidence estimation methods like ours. 
    \item \Uniform: Uniform averaging or majority voting. This is in contrast to the ``weight'' method where weights are determined by the estimated labeler confidence.
\end{itemize}

All methods are applied on four real crowdsourcing datasets obtained from Meta, covering a broad range of different labeling tasks. See basic information of the datasets in Tab.~\ref{tab:datasets}. For example, in the \TOPIC dataset, each annotation is a subset of 26 predefined topics. The
similarity of two annotations is their Jaccard similarity and aggregation is performed by majority voting on each topic. 

Note that to use the \POAK algorithm we need to somehow partition annotations into types. In all datasets we used a simple straightforward partition. E.g. in the \TOPIC dataset we assign each singleton to a type and group all non-singletons into a separate type. 

Comprehensive descriptions and partitioning details are included in the Appendix.
%App.~\ref{apx:emp}. 

\newpar{Evaluation}
We first randomly split each dataset into a training set and a test set. We then run the learning component (Alg.~\ref{alg:OAK}) of each variant to estimate labeler's confidence. Given the learned confidence, we implement the estimation component (Alg.~\ref{alg:OAK_predict}) to estimate the accuracy of the first label for each job in the test set, only inquiring the next label if the accuracy is below a given threshold. At each threshold, we calculate the cost and accuracy by averaging over all jobs in the test set, which corresponds to a point on the curve in Fig. \ref{fig:dmr_example}. By varying the threshold, we are able to generate a cost-accuracy curve for each method as in Fig. \ref{fig:dmr_example}.
 
 This cost-accuracy curve of each algorithm is compared against the \Uniform baseline which uses a biased coin-flip to decide whether to use one or all annotations. We evaluate the performance of all methods by computing the \emph{Relative Area Under the Curve} (RAUC), represented as the shaded area in the illustrating example in Fig.~\ref{fig:dmr_example}.

\newsubsec{Results}
\newpar{Comparison with the baselines} Fig.~\ref{fig:results} shows the RAUC results of the different methods on the four datasets.

The confidence-based methods \POAK-Weight and \OAK-Weight consistently outperform non-model based alternatives, \SAD and \Uniform, which justifies the usefulness of confidence estimation in improving the cost-quality trade-off curve. Averaging or majority voting proves to be more effective than selection across most of the settings. This finding is in line with a similar conclusion from \citet{Braylan2021aggregating}. 

In addition, \POAK outperforms \OAK with \POAK-Weight dominating all other methods. This implies that the competencies of labelers vary across different annotation types. The partition-based method can effectively learn and adjust for this heterogeneity. Finally, as the training sample size increases, the performance of all methods improves as expected as the methods can make more accurate estimations of labeler confidence.

\newpar{Deep dive on \POAK}
As \POAK-Weight dominates all other alternatives, we dive deeper into the different variants of \POAK\footnote{For simplicity, the term ``Weight'' is omitted from names hereafter, given that all \POAK variants use this aggregation method.}, including \POAKi and \POAK-IRT, to examine the effectiveness of calibration in confidence estimation. 
It turns out that data is not uniformly distributed across different annotation types. We tackle this by computing the \POAKi estimation \emph{in addition}, using it as a baseline for smoothing instead of \OAK. \xc{I am not exactly following this part. Did a few changes but not fully sure whether it makes sense.} As shown in Fig.~\ref{fig:regularization}, both \POAKi and \POAK-IRT outperform \POAK, highlighting the value of calibration particularly when training sample is small. In addition, IRT-based regularization (purple curve) further improves over \POAKi by allowing for information sharing across partitions to prevent overfitting. 

In addition to the performance evaluation, we also examine the predictive accuracy of \POAK. Specifically, the predicted accuracy of labels and their true accuracy are highly correlated (Fig.~\ref{fig:cal_scatter}). Furthermore, the estimation accuracy exhibits heterogeneity across different annotation types.

% \begin{itemize}
%     \item The confidence-based methods \POAK-Weight and \OAK-Weight consistently outperform non-model based alternatives \SAD and \Uniform, which justifies the usefulness of confidence estimation.
%     \item Averaging or majority voting proves to be more effective than selection across most of the settings. This finding is in line with a similar conclusion from \citet{Braylan2021aggregating}. 
%     \item \POAK outperforms \OAK with \POAK-Weight dominating all other methods. This implies that the competencies of labelers vary across different annotation types. The partition-based method can effectively learn and adjust for this heterogeneity. 
%     \item As the training sample size increases, the performance of all methods improves as expected as the methods can make more accurate estimations of labeler confidence.
% \end{itemize}

%The most important parameter that affects performance is the size of the training set. Clearly with more data we expect better performance. Also keep in mind that the \POAK variant uses substantially more latent variables than the other variants and is thus more prone to overfitting. We thus expect it to outperform the basic variant only when there is sufficient training data.
 
%how much of the cost could be saved by taking a threshold that still guarantees almost-perfect accuracy (see right figure). The figures below show the results of all algorithms, as we increase the size of the training data.

\begin{figure}[t]
    \centering
    \includegraphics[width=0.96\linewidth]{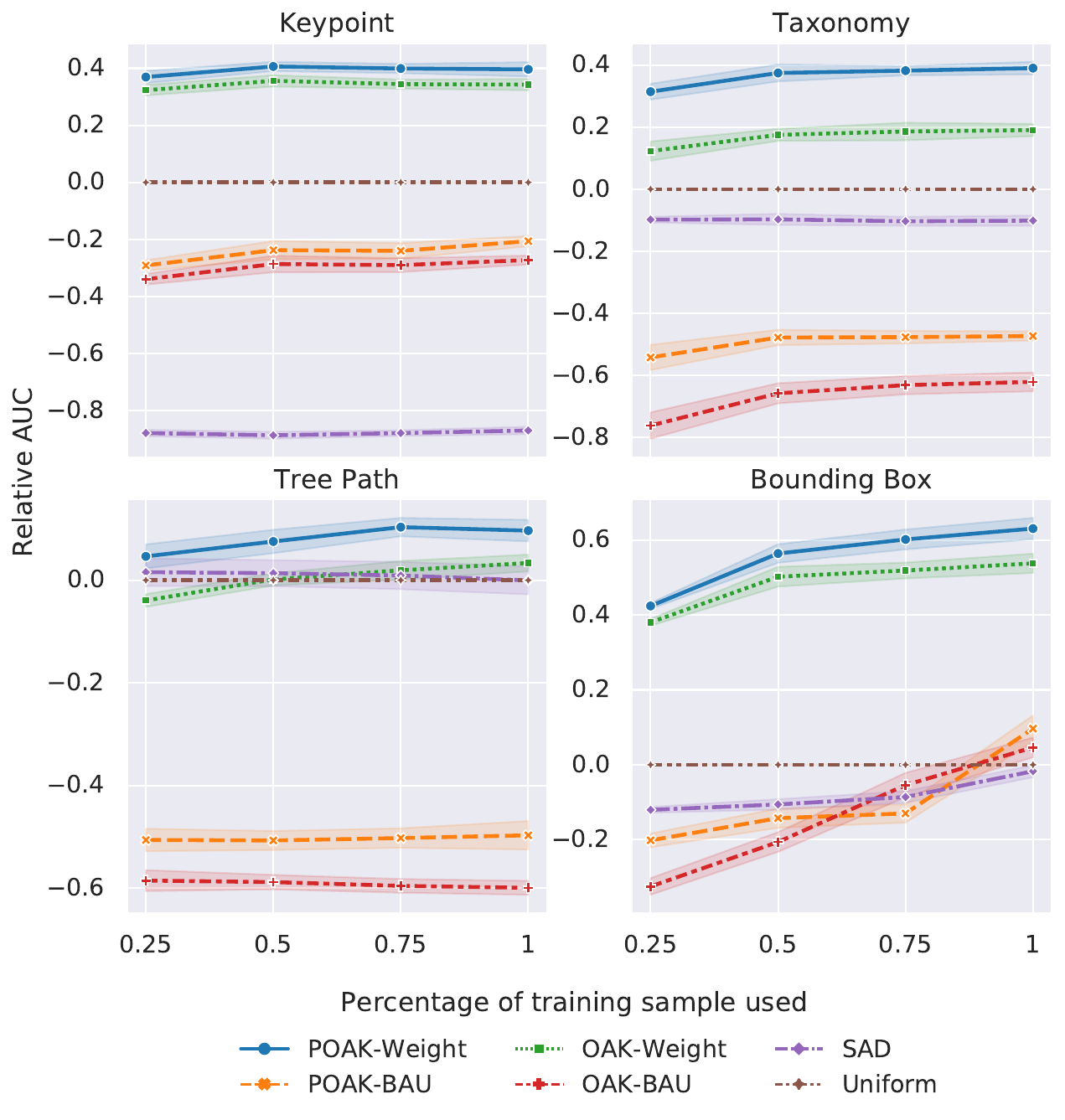}
    \caption{RAUC results (relative to \Uniform) of all methods on four datasets. Point estimates and 95\% confidence intervals are obtained over 10 trails under each setting.}
    \vspace{-0mm}
    \label{fig:results}
\end{figure}

\begin{figure}[t]
    \centering
    \includegraphics[width=0.92\linewidth]{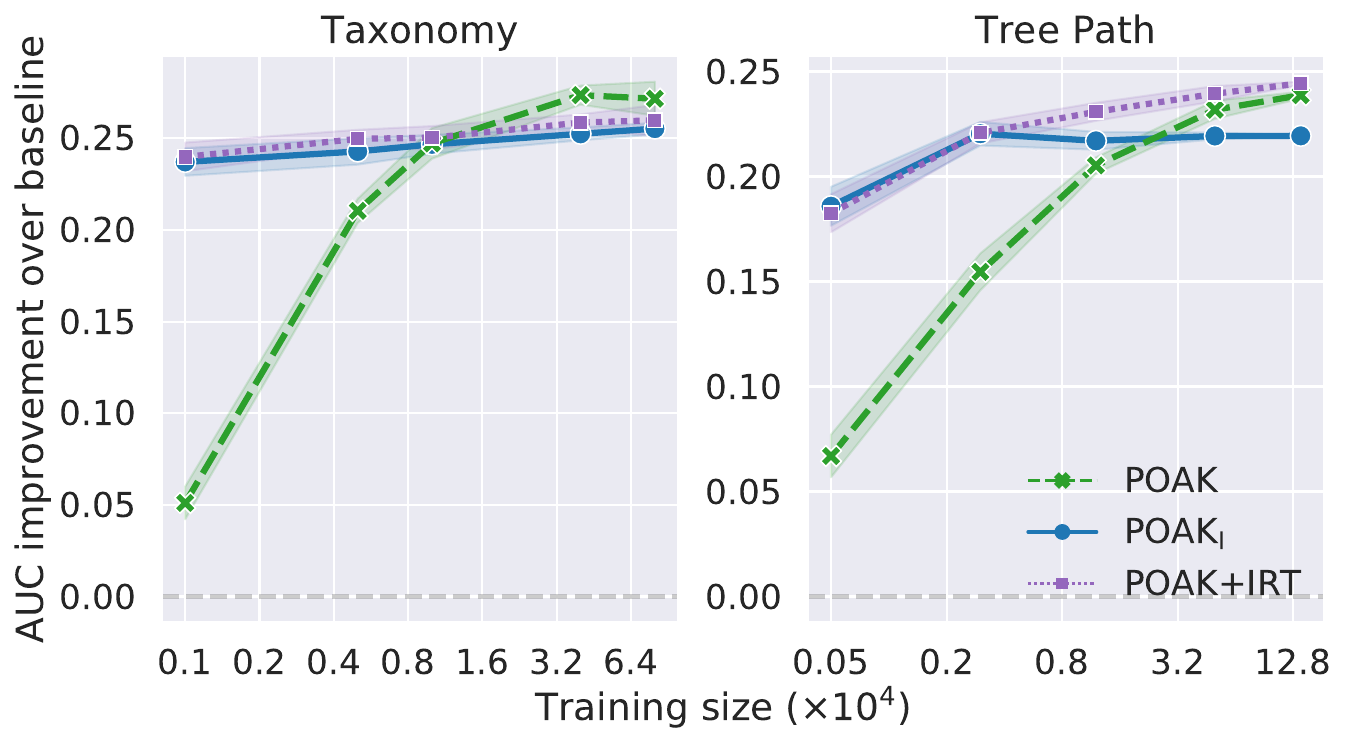}
    \caption{Comparison between different \POAK variants.}
    \label{fig:regularization}
\end{figure}

%We can see in Fig.~\ref{fig:sample_size} that all three variants of the \OAK algorithm  improve performance substantially. When there is sufficient data, the partition variant is (unsurprisingly) the best, but it may overfit when there is not enough data, especially in domains with many types. 
% \newpar{Overfitting and smoothing}

\begin{figure}[t]
\includegraphics[width=\linewidth]{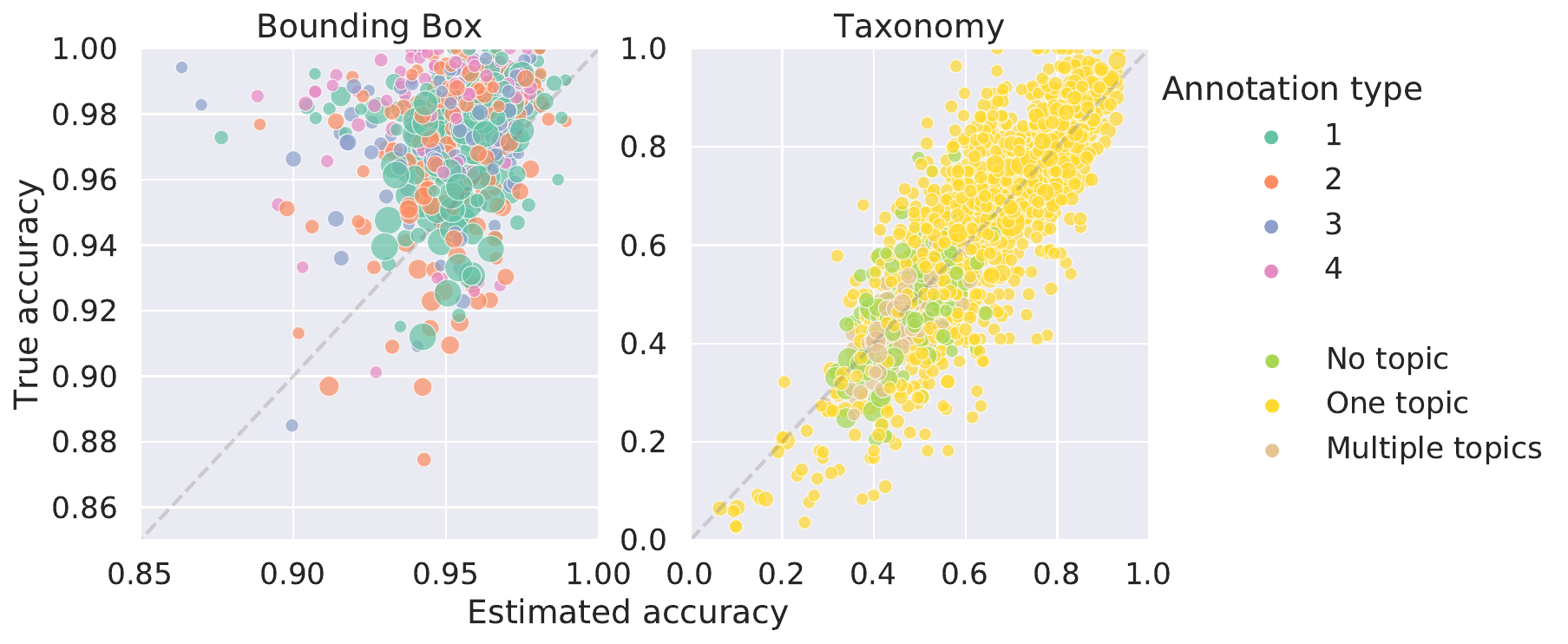}
\caption{ A plot of estimated accuracy $\hat c^{(\ell)}_i$ vs. actual accuracy computed over all items in the test set. Each point represents a pair $(i,\ell)$ of worker and label type, where larger dots represent pairs with more samples in the data. The enlarged version of the left panel is in the Appendix.}
\label{fig:cal_scatter}
\end{figure}

\newsec{Conclusion}
\label{sec:conclusion}
In this paper, we develop novel modeling approaches to improve the efficiency of online crowdsourcing processes with complex annotations. The models proposed are task-independent and applicable to any complex annotation tasks in which the pairwise similarity between two arbitrary annotations can be defined. These models are based on the underlying Anna Karenina principle that good workers are similar to one another in their reported annotations.

We first extend previous work on \PTD  to propose \OAK in the online setting and then introduce two extensions: (1) \POAK which estimates the accuracy of complex annotations by first partitioning the observed annotations into types and then applying \OAK to estimate workers' per-type competence, and (2) \POAKi that reduces the number of parameters  by using Item Response Theory.

We provide theoretical proofs that the Anna Karenina principle extends to per-reported-type estimations, which generalizes the results of~\citet{meir2023easy}. We also provide extensive empirical results comparing the effectiveness of our methods on four real-world applications.

\section*{Acknowledgements}
We appreciate the helpful comments of the anonymous reviewers. Reshef Meir was partly funded by the Israel Science Foundation (ISF grant 2539/20).

\bibliography{proxy,clara}

\newpage
\appendix
\onecolumn

\section{Semi-supervised learning}
\label{apx:supervised}

While it is true that $\hat c^0_i$ by itself is unreliable for most agents (due to few samples) we would still not want to throw away this information. As we now have two separate (albeit not completely independent) estimations of $c_i$ (namely $\hat c^0_i$ and $L(\pi_i)$), the most natural thing is to combine them together. We use a weighted average of the two estimations, with a meta-parameter $\alpha^{}\geq 1$ specifying how many pairwise samples are equal to each auditor label:
\begin{equation}
    \label{aeq:c_hat}
    \hat c_i := \frac{\alpha^{} m^*_i}{\alpha^{} m^*_i+ m_i}\hat c^0_i +  \frac{m_i}{\alpha^{} m^*_i+ m_i} L(\pi_i).
\end{equation}

Since both $L(\pi_i)$ and $\hat c^0_i$ are unbiased estimators of $c_i$, so is $\hat c_i$.

Note that while this substantially improves performance (when we have auditor labels) and is recommended in practice, in this paper we did not include this step in our simulations, since our focus is on the benefit we can get from average similarity.  

\rmr{Ideally we would add here figures showing the combined effect of OAK / POAK with auditor labels}

\section{Proof of the Conditional Anna Karenina theorem}
\label{apx:POAK}

We make the two following observations:
\begin{small}
\begin{align}
      c_i^{(\ell)} & = Pr[z^*=\ell| x_i=\ell] = \frac{Pr[z^*=\ell, x_i=\ell]}{Pr[x_i=\ell]}\notag\\
      &= \frac{q^{(\ell)} p_i^{(\ell)}}{Pr[x_i=\ell]}\label{aeq:cil}\\
    Pr[x_i=\ell] & = Pr[x_i=\ell,z^*=\ell] + \sum_{\tau\neq \ell}Pr[x_i=\ell,z^*=\tau] \notag\\
    &= q^{(\ell)} p_i^{(\ell)}+ \sum_{\tau\neq \ell}q^{(\tau)} \mat_i^{\tau\rightarrow \ell} \label{aeq:pxil}
\end{align}
\end{small}

\paragraph{A conditional Anna Karenina principle for categorical labels}

\begin{theorem}\label{thm:AK_cond} Fix prior probabilities $q$, a category $\ell$, and a worker $i$ with partial type $\mat^{(-\ell)}_i$. Then there are constants $\alpha^{(\ell)}, \beta^{(\ell)}$ such that
$E[\pi_i^{(\ell)}]=\alpha^{(\ell)} c_i^{(\ell)} + \beta^{(\ell)}$.
\end{theorem}
\begin{lemma}\label{lemma:YZ}
For any $Z,Y>0$, 
$$\frac{1}{Z+Y} = \frac{-1}{Y} \cdot \frac{Z}{Z+Y} + \frac{1}{Y}.$$
\end{lemma}
\begin{proof}
 \begin{align*}
 &\frac{-1}{Y} \cdot \frac{Z}{Z+Y} + \frac{1}{Y} =  \frac{-1}{Y} \frac{Z}{Z+Y} + \frac{Z+Y}{Y(Z+Y)}\\
 &= \frac{Z+Y-Z}{Y(Z+Y)} = \frac{Y}{Y(Z+Y)}=\frac{1}{Z+Y}.\qedhere
 \end{align*}
\end{proof}
\xc{Do we need to write this as a lemma?}
\begin{proof}[Proof of Theorem~\ref{thm:AK_cond}] We will write $Pr[x_{i'}=\ell | x_i=\ell]$ explicitly as a function of the type $M_i$, then split it into two terms and show that each of them is linear in $c_i^{(\ell)}$.
\begin{small}
\begin{align}
  Pr&[x_{i'}=\ell|x_i=\ell] = Pr[x_{i'}=\ell|x_i=\ell,z^*=\ell]Pr[z^*=\ell|x_i=\ell] \notag\\
  &~~~+ \sum_{\tau\neq \ell}Pr[x_{i'}=\ell|x_i=\ell,z^*=\tau]Pr[z^*=\tau|x_i=\ell]\notag\\
       &= Pr[x_{i'}=\ell|z^*=\ell] c_i^{(\ell)} + \sum_{\tau\neq \ell}Pr[x_{i'}=\ell|z^*=\tau]Pr[z^*=\tau|x_i=\ell]\notag\\
              &= p_{i'}^{(\ell)} c_i^{(\ell)} + \sum_{\tau\neq \ell}Pr[x_{i'}\!=\!\ell|z^*\!\!=\!\tau]Pr[z^*\!\!=\!\tau|x_i\!=\!\ell]\notag\\
                   &= p_{i'}^{(\ell)} c_i^{(\ell)} + \sum_{\tau\neq \ell}M_{i'}^{\tau\rightarrow \ell}\frac{q^{(\tau)} M_i^{\tau\rightarrow \ell}}{Pr[x_i=\ell]} \\
                   &= p^{(\ell)}_{i'} c_i^{(\ell)} + \frac1{Pr[x_i=\ell]}\sum_{\tau\neq \ell}q^{(\tau)} M_i^{\tau\rightarrow \ell} M_{i'}^{\tau\rightarrow \ell} \label{aeq:sbb}
\end{align}
\end{small}
Eq.~\eqref{aeq:sbb} is a sum of two terms. The first one is obviously linear in $c_i^{(\ell)}$ since $\mat_{i'}$ is fixed. For the second term, we first observe that each $q^{(\tau)} M_i^{\tau\rightarrow \ell} M_{i'}^{\tau\rightarrow \ell}$ for $\tau\neq \ell$ is completely determined by the fixed terms $q,\mat_{i'}$ and $\mat^{(-\ell)}_i$.
It remains to show that $\frac1{Pr[x_i=\ell]}$ is linear in $c_i^{(\ell)}$.

For this, denote $Z:= q^{(\ell)}p_i^{(\ell)}$; $Y= \sum_{\tau\neq \ell}q^{(\tau)}\mat_i^{\tau\rightarrow\ell}$, and note that $Y$ is a constant. Then Eqs.\eqref{aeq:cil},\eqref{aeq:pxil} can be written as:
$$Pr[x_i=\ell] = Z + Y;~~~~c_i^{(\ell)}=\frac{Z}{Z+Y}.$$

By Lemma~\ref{lemma:YZ}, 
$$\frac1{Pr[x_i=\ell]} = \frac{1}{Z+Y} =  \frac{-1}{Y} \cdot \frac{Z}{Z+Y} + \frac{1}{Y} =  \frac{-1}{Y} \cdot c_i^{(\ell)} + \frac{1}{Y},$$
so indeed this part of the expression is linear as well. We now join the two parts to find the coefficients of the linear function: 
\begin{small}
\begin{align*}
     Pr&[x_{i'}=\ell|x_i=\ell]  = p^{(\ell)}_{i'} c_i^{(\ell)} + \frac1{Pr[x_i=\ell]}\sum_{\tau\neq \ell}q^{(\tau)} M_i^{\tau\rightarrow \ell} M_{i'}^{\tau\rightarrow \ell}\\
     &=  p^{(\ell)}_{i'} c_i^{(\ell)} + ( \frac{-1}{Y} \cdot c_i^{(\ell)} + \frac{1}{Y}) \sum_{\tau\neq \ell}q^{(\tau)} M_i^{\tau\rightarrow \ell} M_{i'}^{\tau\rightarrow \ell}\\
     &= \left( p^{(\ell)}_{i'} - \frac{\sum_{\tau\neq \ell}q^{(\tau)} M_i^{\tau\rightarrow \ell}M_{i'}^{\tau\rightarrow \ell}}{\sum_{\tau\neq \ell}q^{(\tau)}\mat_i^{\tau\rightarrow\ell}}\right)c_i^{(\ell)} \\
     &~~~+  \frac{\sum_{\tau\neq \ell}q^{(\tau)} M_i^{\tau\rightarrow \ell}M_{i'}^{\tau\rightarrow \ell}}{\sum_{\tau\neq \ell}q^{(\tau)}\mat_i^{\tau\rightarrow\ell}}, 
\end{align*}
\end{small}
so $\alpha^{(\ell)}_{i'} = p^{(\ell)}_{i'} - \frac{\sum_{\tau\neq \ell}q^{(\tau)} M_i^{\tau\rightarrow \ell}M_{i'}^{\tau\rightarrow \ell}}{\sum_{\tau\neq \ell}q^{(\tau)}\mat_i^{\tau\rightarrow\ell}}$ and $\beta^{(\ell)}_{i'} = \frac{\sum_{\tau\neq \ell}q^{(\tau)} M_i^{\tau\rightarrow \ell}M_{i'}^{\tau\rightarrow \ell}}{\sum_{\tau\neq \ell}q^{(\tau)}\mat_i^{\tau\rightarrow\ell}}$.
\end{proof}
Now, as in \cite{meir2023easy}, we are ultimately interested in the expected [conditional] similarity to a \emph{random worker}, which is still linear in the [conditional] competence. When sampling a random type $\mat_{i'}$ from an arbitrary distribution, we get that $\alpha^{(\ell)}_{i'},\beta^{(\ell)}_{i'}$ are random variables. We denote their expectations by $\alpha^{(\ell)},\beta^{(\ell)}$, respectively.  
\begin{corollary}
 $E[s(x_i,x_{i})|x_i=\ell] = \alpha^{(\ell)} \cdot c_i^{(\ell)} + \beta^{(\ell)}$.
 \end{corollary}
 This follows immediately from linearity of (conditional) expectation. 
% \begin{proof}
% \begin{align*}
%     E[s(x_i,x_{i})|x_i=\ell]& = E_{\mat_{i'}}E_{x_{i'}}[s(x_i,x_{i'})|x_i=\ell] =  E_{\mat_{i'}}Pr_{x_{i'}}[s(x_i,x_{i'})=1|x_i=\ell] \\
% & =  E_{\mat_{i'}}Pr[x_{i'}=\ell|x_i=\ell] =  E_{\mat_{i'}}[ \alpha^{(\ell)}_{i'} c_i^{(\ell)} + \beta^{(\ell)}_{i'}] = E_{\mat_{i'}}[ \alpha^{(\ell)}_{i'}] c_i^{(\ell)} + E_{\mat_{i'}}[\beta^{(\ell)}_{i'}]\\
% &= \alpha^{(\ell)} \cdot c_i^{(\ell)} + \beta^{(\ell)}.
% \end{align*}
% \end{proof}
\begin{remark}
The higher $\alpha^{(\ell)}$ is w.r.t $\beta^{(\ell)}$, the better we can separate good workers from poor ones. Note that a sufficient condition for $\alpha^{(\ell)}$ to be positive, is that $E_{\mat_{i'}}[p_{i'}^{(\ell)}]$ is higher than $E_{\mat_{i'}}[\mat_{i'}^{\tau\rightarrow \ell}]$ for all $\tau\neq \ell$. I.e., that the overall competence to identify $\ell$ is higher than the overall chance to incorrectly report a label as $\ell$.  \rmr{add proof?}
\end{remark}

\subsection{Implementation Notes}

\an{We might want to move all the implementation-related details to Section 6 on Experiment}

\an{Update the algorithm name: \POAK}

Formally, denote by $M_i^{(\ell)}:= \{j\in M: x_{ij}\in \calX^{(\ell)}\}$; $M_{ii'}^{(\ell)}:= M_i^{(\ell)} \cap M_{i'}$; and $M^{(\ell),*}_i:= M^{(\ell)} \cap M^*_i$.

Our \POAK algorithm essentially runs Alg.~\ref{alg:OAK} once for every label type $\ell$, each time using only the part of the dataset where $i$ reported a type~$\ell$ label (thus replacing $M_{ii'}$ and $M^*_i$ sets in the Alg.~\ref{alg:OAK} with the respective $M^{(\ell)}_{ii'}$ and $M^{(\ell),*}_i$), and obtain $\Theta = (\hat c^{(\ell)},m_i^{(\ell)})_{i\in N, \ell\leq k}$.

Similarly, the estimation component of \POAK uses $\hat c^{(\ell)}_i$ whenever $x_{ij}\in \calX^{(\ell)}$. Note that even though the theorem is stated only for categorical labels, our empirical results show that the linear relation holds approximately for other labels, when partitioned into arbitrary categories.

\paragraph{Smoothing}
One caveat is that even if the total amount of training data is quite large, some $(i,\ell)$ pairs may still have only few samples, or none at all. We therefore also calculate the average values $\hat c_i$ and $\hat c^{(\ell)}$ and use them as default values when there are not enough samples (so they factor in the calculation of $\hat C_j^{(\ell)}$ in the estimation component). See Alg.~\ref{alg:POAK_predict}

% \begin{algorithm}[t]
% \caption{\textsc{\POAK (estimation Component)}}
% \label{alg:POAK_predict}
% \KwIn{Labels $(x_{ij})_{i\in I_j}$, model $\Theta= (\Theta^{(\ell)})_{\ell\leq k}$}
% \KwOut{Estimated label $\hat z_j$, estimated accuracy $\hat C_j$} 

% Aggregate labels $\hat z_j := agg((x_{ij})_{i\in I_j})$\;

% Set $\ell$ such that $\hat z_j \in \calX^{(\ell)}$\;

% Find closest worker $i^* := \argmax_{i\in N_j}s(x_{ij},\hat z_j)$\;

% Calculate $\hat C_j := \frac{\overline m^{(\ell)}_{i^*}}{\overline m^{(\ell)}_{i^*}+\gamma}\hat c^{(\ell)}_{i^*} + \frac{\gamma}{\overline m^{(\ell)}_{i^*}+\gamma}\hat c_i$\;

% \Return $(\hat z_j, \hat C_j)$\;
% \end{algorithm}

%%%%%%%%%%%%%%%%%%%%%%%%%%%%%%
\section{Reducing Model Size with Item Response Theory}
\label{apx:IRT}

\rmr{If lack space perhaps move all IRT to appendix}
There is an inherent tradeoff in the size of the model, where larger models can catch finer patterns (e.g. accuracy that is label-specific) but also more prone to overfitting, and in general require more data to work well. 

One way to still take label-specific accuracy into account with a lower number of latent variables, is to impose independence between worker competence and the features of each label. Thus we can maintain a number of latent variables that is $O(n+k)$ rather than $O(n\times k)$.  The tricky part is that it is not a-priori obvious what label-specific features should be learned, and how to combine them with the workers' estimated competence. Fortunately, there is a rich theory focused on a very similar question, namely \emph{Item Response Theory}~\cite{hambleton1991fundamentals}.

\subsection{IRT Background} 
\label{sec:irt_background}

The original purpose of IRT is not machine learning or truth discovery, but rather the creation and evaluation of \emph{standardized tests}.

 IRT assumes each `item' (say, question in a test) has three key properties: difficulty, separation, and base rate; and that each responder has some level of competence. 
 
 Intuitively, the probability of a correct response depends on the above four elements as follows:
 \begin{itemize}
     \item The base-rate is the probability of a complete guess to be correct (e.g. $0.25$ in a 4-category multiple choice question);
     \item The probability is increasing with the competence and decreasing with the difficulty;
     \item High separation means a larger difference between competent and incompetent responders. 
 \end{itemize}
 While one could think of many ways to combine the four parameters into a probabilistic formula, the standard IRT formula has the following form:

 \begin{equation}
     \label{eq:IRT}\Pirt(x_{ij})=p^0_j+ \frac{1-p^0_j}{1+\exp(b_j(d_j-c_i))},
 \end{equation}
 where $d_j, b_j, p^0_j$ correspond to item~$j$'s difficulty, separation, and base rate; and $\Pirt$ is the probability that answer $x_{ij}$ is correct. Note that the range of $d_j, c_i$ is the entire real line, and that $\Pirt$ as a function of $d_j-c_i$ has an increasing logit form.
 
A level-2 and level-1 IRT models are simplified versions where the base-rate or the separation or both are assumed to be identical for all items.
 
 Note that if item parameters are known (e.g. a test was written in a particular way), then responders' parameters can be estimated, and it is possible that people who answered correctly fewer questions turn out to be more competent (as they answered harder and/or more separating questions). 
 
 It is more difficult to estimate all parameters together~\cite{BakerBook2004item}, but note that this task as well requires as input the correct answers (or at least which responders correctly answered each question), and typically many responders per questions whereas we have at most 3 and sometimes just 1.

%  \paragraph{Applying IRT to truth discovery}
 
%  We use the IRT formula, except that (as in the multi-parameter extension) we attach the `item' parameters not to specific questions but to an annotation type. To learn parameter values, we first learn the multi-parameter model, and then try to fit it with the (fewer) IRT parameters.

\subsection{IRT Model}
For ease of exposition, suppose we deal with categorical questions, as in the original motivation of IRT (we later explain how to very easily generalize this approach). 

The key assumption we make in this section, is that the probability of a correct answer follows Eq.~\eqref{eq:IRT}.

However, in contrast to the way IRT is used in standardized tests, we attach the three `item' parameters  not to questions but to \emph{reported labels}. Thus the $\Pirt$ in our model is the estimated probability that $x_i$ is correct, conditional on $x_i=\ell$.

Clearly not every model (and not even all conditional accuracy parameters $c^{(\ell)}_i$) can be captured by IRT parameters, simply since there are only $n+3k$ of those. However we next show that if each worker is characterized by a single accuracy parameter (as in most of the truth discovery literature), then a level-1 IRT model can capture both accuracy and priors over labels. 

Note that in a accuracy-only model, we have $\mat_i^{\ell\rightarrow \ell}=p_i$ for all $\ell$, and $\mat_i^{\tau\rightarrow \ell}=(1-p_i)/(k-1)$ for all $\tau \neq \ell$.

 % A \emph{decomposable} noise model assumes that the $n\times k\times k$ parameters of the model are essentially determined by $n+k$ parameters: the accuracy $p_i$ of each worker and the \emph{bias} $b^{(\ell)}$ towards each label. Intuitively, Given a true label $z^*$, worker~$i$ knows the correct label w.p. $p_i$. With the complement probability $1-p_i$ the worker does not know the answer and guesses each label $\ell$ w.p. $b^{(\ell)}$ (note that $b$ may or may not equal the prior over true labels). Formally, in a decomposable model $(a,b)$ we have $\mat_i^{\ell\rightarrow\ell}=p_i + (1-p_i)b^{(\ell)} $
 
 % so that $\mat_i^{\ell\rightarrow \ell}=p_i\cdot e^{(\ell)}$. 

\begin{theorem}
An accuracy-only noise model with arbitrary priors is captured exactly by a level-1 IRT model. 
\end{theorem}
\begin{proof}
We define the logit function $g(x):=\frac{1}{1+exp(x)}$. Note that in a level-1 IRT model, Eq.~\eqref{eq:IRT} reduces to $\Pirt(x_{i}^{(\ell)}) = g(d^{(\ell)}-c_i)$.

Am accuracy-only model is given by $n+k$ parameters: worker accuracies $(p_i)_{i\leq n }$, and label priors $(q^{(\ell)})_{\ell \leq k}$. 

We define the IRT parameters as follows:
\begin{itemize}
    \item $c_i := -\log(\frac{1}{p_i}-1)$; and
    \item $d^{(\ell)}:= \log((\frac{1}{q^{(\ell)}}-1)\frac{1}{k-1})$.
\end{itemize}

It is left to show that the original decomposable model and the IRT model always provide the same estimation, i.e. that 
$$g(d^{(\ell)} - c_i) = Pr[z^*=\ell |x_i = \ell] .$$

To that end:
\begin{align*}
&Pr[z^*=\ell |x_i = \ell]= \frac{q^{(\ell)}p_i}{Pr[x_i=\ell]}= \frac{q^{(\ell)}p_i}{\sum_{\ell'} q^{(\ell')}Pr[x_i=\ell|z^*=\ell']}\\
&=  \frac{q^{(\ell)}p_i}{q^{(\ell)}p_i+\sum_{\ell'} q^{(\ell')}\mat_i^{\ell'\rightarrow \ell}}=  \frac{q^{(\ell)}p_i}{q^{(\ell)}p_i+\sum_{\ell'\neq \ell} q^{(\ell')}\frac{1-p_i}{k-1}}\\
&= \frac{q^{(\ell)}p_i}{q^{(\ell)}p_i+\frac{1-p_i}{k-1}\sum_{\ell'\neq \ell} q^{(\ell')}}
= \frac{q^{(\ell)}p_i}{q^{(\ell)}p_i+\frac{1-p_i}{k-1}(1- q^{(\ell)})} \\
&=\frac{1}{1+\frac{1}{k-1}(1-p_i)(1- q^{(\ell)})/q^{(\ell)}p_i}\\
&= \frac{1}{1+\frac{1}{k-1}\frac{1-p_i}{p_i} \frac{1- q^{(\ell)}}{q^{(\ell)}}}
\end{align*}
Where on the other hand,
\begin{align*}
    \Pirt&(x_i^{(\ell)}) = g(d^{(\ell)}-c_i) = \frac{1}{1+\exp(d^{(\ell)}-c_i)} \\
    &= \frac{1}{1+\exp(\log((\frac{1}{q^{(\ell)}}-1)\frac{1}{k-1})+\log(\frac{1}{p_i}-1))} \\
    &= \frac{1}{1+\exp(\log((\frac{1}{q^{(\ell)}}-1)\frac{1}{k-1}(\frac{1}{p_i}-1)))} \\
    &= \frac{1}{1+(\frac{1}{q^{(\ell)}}-1)\frac{1}{k-1}(\frac{1}{p_i}-1)} \\
    &= \frac{1}{1+\frac{1}{k-1}\frac{1-p_i}{p_i} \frac{1- q^{(\ell)}}{q^{(\ell)}}}=Pr[z^*=\ell |x_i = \ell],
\end{align*}
as required.
\end{proof}
\rmr{add proof}
\subsection{Implementation}
We implemented a third variant of our \OAK algorithm, called \POAK-IRT (or $\POAKi$) as follows. In the learning component, 
 the base rates $p_0^{(\ell)}$ are computed directly from the data. 
We then learn a full $n\times k$ model $\Theta$ using the learning component of $\POAK$. Finally, we fit the remaining $2k+n$ latent variables of a level-3 IRT model $\Theta^I$ to the accuracy matrix of $\Theta$. There are many ways to do so, but we use several steps of linear regression,  

In the estimation component, the \POAKi algorithm estimates accuracy according to Eq.~\eqref{eq:IRT}.

\subsection{\POAK+IRT}
Due to the obvious generalization tradeoff involved in restricting the parameter space, the \POAKi  algorithm reaches good performance much faster (i.e. with a smaller training set) than \POAK, but its performance is bounded. 

The \POAK+IRT algorithm combines the two by using the \POAK estimation, with \POAKi estimation as regularizer with a fixed weight for linear smoothing.  Thus intially the \POAKi estimation is used and for categories on which there is ample training data, \POAK kicks in for higher accuracy. 

This is the algorithm used in Fig.~\ref{fig:regularization}
\rmr{explain application to general labels}
%%%%%%%%%%%%%%%%%%%

\section{Multiple Decision Points}\label{apx:acc_aggregated}
In the main text we only explained how to use our accuracy estimation algorithm at a single decision point---after obtaining the first annotation of each item.

Of course, after obtaining a second annotation, we could also try and decide whether this is sufficient or a third one is needed (if exists). However naively applying the same algorithm is a problem since it is set to evaluate \emph{a single annotation by a single labeler}. Clearly, there is relevant information in all available annotations.

For example, suppose that we collected two labels for item $j$, namely $x_{j1}, x_{j2}$ that are identical but each has a very low accuracy estimate (either because the reported label has a low prior or because both workers are often inaccurate). The fact that the labels are identical though should also be taken into account, and perhaps to a lesser extent if they are only similar. 

Intuitively, we want to `boost' the accuracy estimate of aggregated labels if they are similar, on average, to all the reported labels, Likewise we would like to penalize the accuracy estimation if reported labels strongly diverge.

To that end, we add a small number of latent variables to the model, two for each decision point: $(\delta_i,\epsilon_i)_{i\leq t}$ where $t=\max_j|N_j|$.

Intuitively, given an aggregated label $\hat z_j=agg(x_{j1},\ldots,x_{ji})$ and its naive accuracy estimate $\hat C_j$ (obtained from Alg.~\ref{alg:OAK_predict}), we use $\delta_i,\epsilon_i$  to apply a non-linear shift on $\hat C_j$.

Formally, we set
$$\tilde{C_j}:=g(g^{-1}(\hat C_j) + \delta_i + \overline s_j \cdot \epsilon_i),$$
where $g(\cdot)$ is the logit function, and $\overline s_j := \frac{1}{i}\sum_{i'=1}^i s(\hat z_j, x_{i'j})$ is the average similarity of the aggregated label $\hat z_j$ to all reported labels. 

Note that this means we apply a linear adjustment of the accuracy estimation in the space of IRT parameters, i.e. we linearly adjust the term $b^{(\ell)}(d^{(\ell)}-c_i)$ before casting it again to $[0,1]$ with the logit function.

The fact that $\delta_i,\epsilon_i$ are a linear transformation (in another space) means we can easily learn them in the learning component by finding an approximate linear transformation from $ \frac{1}{i}\sum_{i'=1}^i s(z^*_j, x_{i'j})$ to the difference $g^{-1}(\hat C_j)-g^{-1}(c^*_j)$ over items with auditor label in the training set.

\paragraph{Multiple thresholds}
We note that estimating accuracy alone is not enough, since we now also have multiple decision thresholds. In contrast to the case of a single decision threshold where typically our algorithms \emph{always} beat the baseline, setting multiple thresholds may hurt performance. 

In Fig.~\ref{fig:2D} we show the result of applying two decision points \emph{with the same threshold} in the two datasets for which we have auditor labels.

We can see that while there is a nontrivial improvement in the Topics dataset, the second decision point hurts performance in the Bounding Box dataset. This may seem strange as we can always decide to \emph{always} request another label at the second decision point (making it moot). However for this we must allow different thresholds. So what the figure shows is that using an identical threshold in both steps can be suboptimal.

\begin{figure*}
    \centering
    \includegraphics[width=0.6\linewidth]{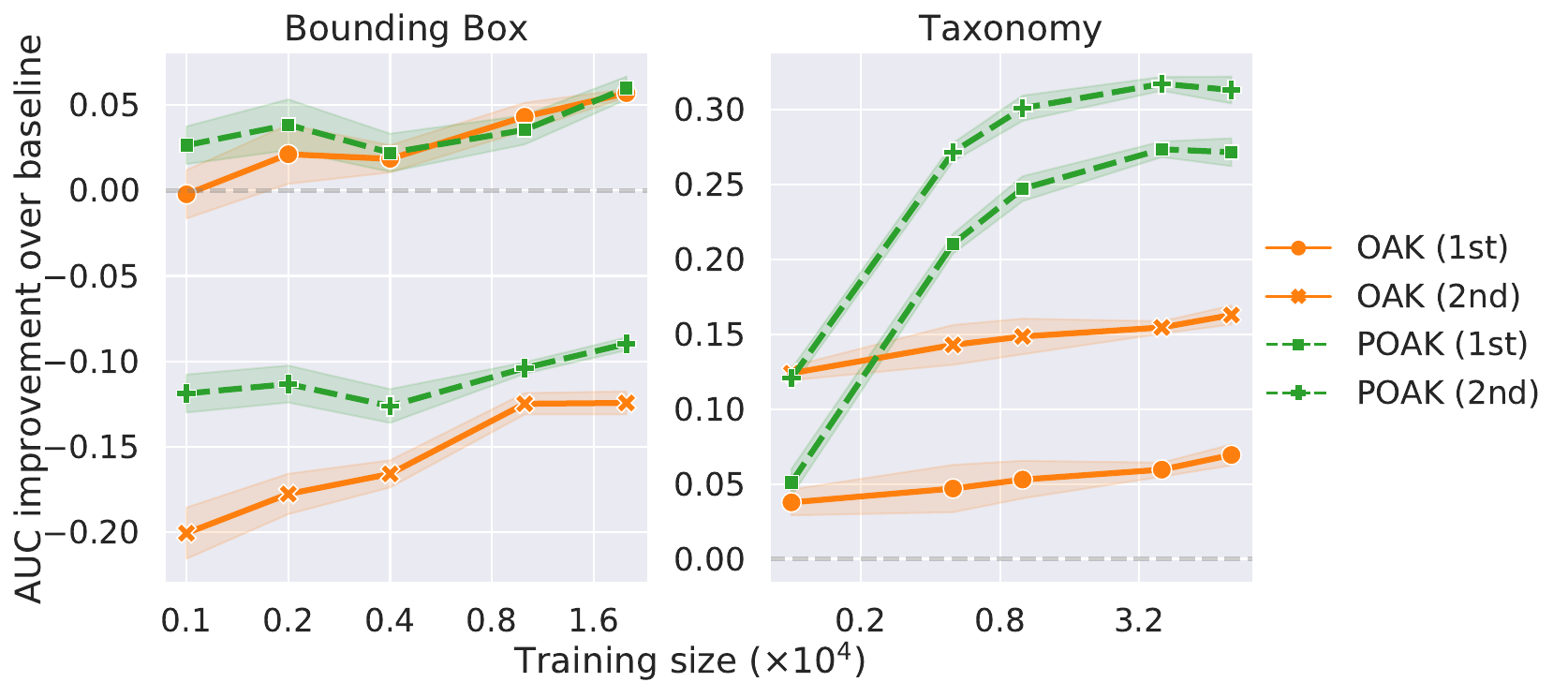} 
    \caption{The effect of adding another decision point \emph{with the same threshold}, after obtaining a second annotation.}
    \label{fig:decision}\label{fig:2D}
\end{figure*}

\section{More Details on Empirical Results}\label{apx:emp}

\paragraph{Datasets} We used four datasets collected at a major social networking platform, with both simple and complex labels.

%\begin{itemize}[leftmargin=0.5cm]
\begin{itemize}
    \item \KEYPOINT: In this dataset, each label is a 2-dimensional coordinate pair to mark a particular body part. Aggregation is performed using (weighted or uniform) average over labels. 
    
    \item \TOPIC: Each annotation is a (possibly empty) subset of 26 predefined topics, which is a part of a large taxonomy. The similarity of two annotations is their Jaccard similarity and aggregation is performed by majority voting on each topic. For annotation partition, we assign each singleton to a type and group all non-singletons into a separate type. 
    
    \item \TREEPATH: Each annotation is a path in a 3-level tree in which the first, second, and leaf levels have $200$, $1500$, and $6000$ nodes, respectively. The similarity measure assigns a score of 1.0, 0.75, 0.5, or 0.0, depending on how the two annotations' first-level nodes differ. Aggregation is by ``hierarchical majority'' starting from the root. For annotation partition we used the first level of the tree.
    
    \item \BOUNDINGBOX: Every annotation is a set of rectangular bounding boxes around people in an underlying image. The similarity of two sets is by the Jaccard similarity of the induced bitmaps, where every bit is `1' if it is contained in some bounding box and `0' otherwise.  Aggregation is by taking a majority vote on every bit. For the annotation partition, we simply used the number of distinct bounding boxes in each annotation, up to 8. However we note that there were very few labels with more than 4 boxes.
\end{itemize}

Unfortunately, all datasets  are proprietary and we have no permission to share them. However the algorithms can be easily reproduced. 
The run-time of all algorithms on the heaviest dataset (20K Bounding Box items, each containing 3 bitmaps of size 200x200) is under 5 minutes.

\paragraph{Calibration}\label{apx:cal}
In datasets where we have some auditor labels, using them for calibration provides a small yet positive effect on accuracy. See Fig.~\ref{fig:calibration}.  Note that this is only true for \POAK. In the simple \OAK algorithm calibration would simply change the threshold required to reach every point on the performance curve, but would not change the curve itself.

\begin{figure}
    \centering
    \includegraphics[width=0.5\linewidth]{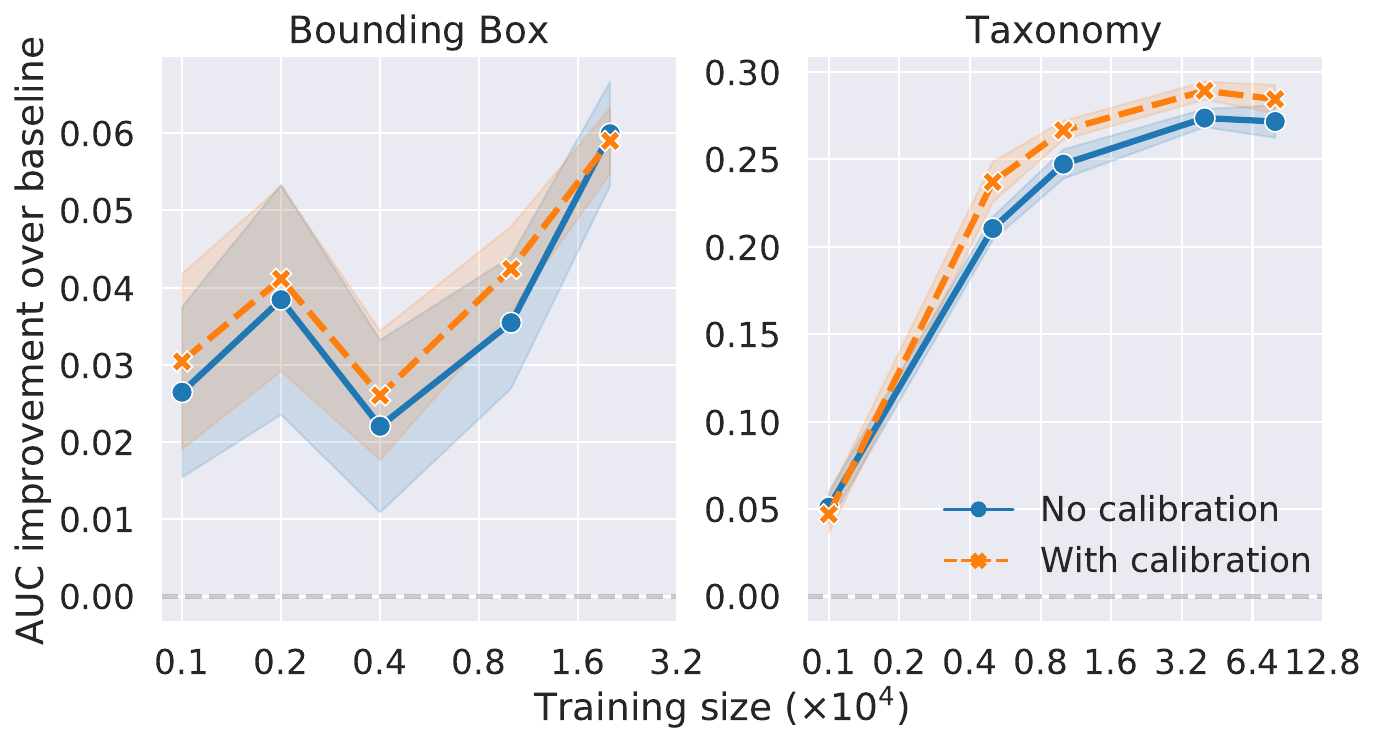}
    \includegraphics[width=0.45\linewidth]{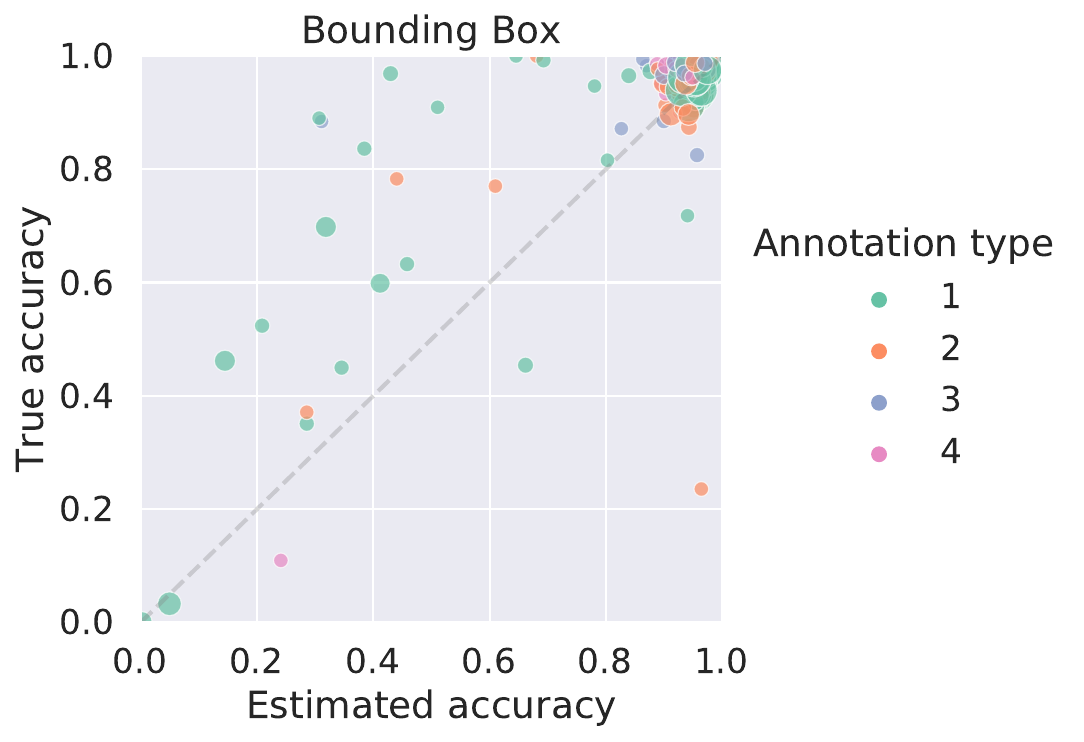}
    \caption{Left: The effect of calibration using available ground truth training data. Right: The full scatterplot of the blowup shown in Fig.~\ref{fig:cal_scatter}.}
    \label{fig:calibration}
\end{figure}

\end{document}